\documentclass{article}

\usepackage{microtype}
\usepackage{graphicx}
\usepackage{subfigure}
\usepackage{float}
\usepackage{longtable}
\usepackage{booktabs} 

\usepackage{hyperref}

\usepackage[accepted]{icml2025}

\usepackage{amsmath}
\usepackage{amssymb}
\usepackage{mathtools}
\usepackage{amsthm}
\usepackage{multirow} 
\usepackage{adjustbox} 

\usepackage[capitalize,noabbrev]{cleveref}

\theoremstyle{plain}
\newtheorem{theorem}{Theorem}[section]

\theoremstyle{definition}

\theoremstyle{remark}

\usepackage[textsize=tiny]{todonotes}

\icmltitlerunning{Learning to Act and Cooperate for Distributed Black-Box Consensus Optimization}

\begin{document}

\twocolumn[
\icmltitle{Learning to Act and Cooperate for Distributed Black-Box Consensus Optimization}




\begin{icmlauthorlist} 
\icmlauthor{Zi-Bo Qin}{scut} 
\icmlauthor{Feng-Feng Wei}{scut} 
\icmlauthor{Tai-You Chen}{scut} 
\icmlauthor{Wei-Neng Chen}{scut} 
\end{icmlauthorlist} 

\icmlaffiliation{scut}{ School of Computer Science and Engineering, South China University of Technology, Guangzhou, China } \icmlcorrespondingauthor{Wei-Neng Chen}{cswnchen@scut.edu.cn} 

\icmlkeywords{Machine Learning, ICML}

\vskip 0.3in
]




\footnotetext[1]{School of Computer Science and Engineering,
South China University of Technology, Guangzhou, China.
Correspondence to: Wei-Neng Chen \texttt{<cschenwn@scut.edu.cn>}.}

\begin{abstract}
Distributed black-box consensus optimization is a fundamental problem in multi-agent systems, where agents must improve a global objective using only local objective queries and limited neighbor communication. Existing methods largely rely on handcrafted update rules and static cooperation patterns, which often struggle to balance local adaptation, global coordination, and communication efficiency in heterogeneous non-convex environments. In this paper, we take an initial step toward trajectory-driven self-design for distributed black-box consensus optimization. We first redesign the agent-level swarm dynamics with an adaptive internal mechanism tailored to decentralized consensus settings, improving the balance between exploration, convergence, and local escape. Built on top of this adaptive execution layer, we propose Learning to Act and Cooperate (LAC-MAS), a trajectory-driven framework in which large language models provide sparse high-level guidance for shaping both agent-internal action behaviors and agent-external cooperation patterns from historical optimization trajectories. We further introduce a phased cognitive scheduling strategy to activate different forms of adaptation in a resource-aware manner. Experiments on standard distributed black-box benchmarks and real-world distributed tasks show that LAC-MAS consistently improves solution quality, convergence efficiency, and communication efficiency over strong baselines, suggesting a practical route from handcrafted distributed coordination toward self-designing multi-agent optimization systems. 
\end{abstract}
\section{Introduction}
\label{submission}
The next generation of intelligent systems will increasingly operate in distributed and networked environments, where multiple agents must make coordinated decisions under local observability, limited communication, and heterogeneous feedback.\cite{chen2025confluenceECMAS,olfati2007consensus} Such settings arise in networked sensing, wireless communication, autonomous coordination, and other large-scale multi-agent systems,\cite{yulu2024joint} where no single agent has access to the full system state or global objective. From the perspective of optimization, these scenarios are closely related to the classical formulation of decentralized optimization under communication constraints.\cite{Nedic2009Distributed} Distributed black-box consensus optimization provides a representative formulation of this challenge: agents can only query local objective values, yet are expected to collectively approach a globally desirable solution while driving the system toward consensus.

Despite substantial progress in distributed optimization, existing approaches remain limited in their ability to support such adaptive decentralized intelligence. Classical gradient-based methods offer strong theoretical guarantees, with representative examples including EXTRA~\cite{shi2015extra} and consensus-based ADMM~\cite{boyd2011admm}, but they rely on explicit objective structure and are often unsuitable for black-box and highly non-convex settings. Reinforcement learning based approaches provide flexibility in handling complex dynamics and partial observability~\cite{zhang2021marl,li2023race,yu2022ppo_marl,li2024incentivize,mcclellan2024marl_generalization,he2022manyagentPO}, yet often suffer from unstable training, weak scalability, and difficult credit assignment in multi-agent environments. Heuristic and swarm-based methods provide a practical alternative for distributed black-box optimization~\cite{chen2025maes,chen2025masoie,chen2025masterMATL}, but most still depend on handcrafted update rules, fixed interaction patterns, and manually designed hyperparameters. As a result, they often struggle to balance local adaptation, global coordination, robustness, and communication efficiency across heterogeneous tasks. More fundamentally, current distributed black-box optimizers remain largely rule-driven, leaving open the broader question of whether multi-agent optimization systems can acquire a degree of self-design capability from historical optimization experience.

In parallel, recent advances in large language models and automated algorithm design have opened a new possibility for optimization: instead of relying solely on manually specified rules, optimization systems may adapt their strategies from historical performance signals, optimization trajectories, and algorithmic feedback~\cite{yang2024opro,wan2024apo,liu2024llambo,ma2024llamoco}. These developments suggest that learning-based mechanisms can serve not only as task solvers, but also as high-level generators of algorithmic behaviors. However, most existing LLM-assisted or learning-driven auto-design methods have been developed for centralized or single-agent settings. Whether such trajectory-driven self-design can be brought into distributed black-box consensus optimization, where agents are limited to local information and neighbor communication, remains largely unexplored. This challenge is especially fundamental in multi-agent systems, because what must be adapted is not only how each agent searches locally, but also how agents coordinate with one another under decentralized constraints. 

Motivated by this gap, we propose \textbf{Learning to Act and Cooperate (LAC-MAS)}, a trajectory-driven collaborative framework for distributed black-box consensus optimization. Our key idea is to endow each agent with two coupled layers of adaptation. At the execution level, we redesign the agent-level swarm dynamics with an adaptive internal mechanism tailored to decentralized consensus settings, improving the balance between exploration, convergence, and local escape. On top of this adaptive execution layer, each agent is further equipped with an LLM that serves as a sparse high-level guidance module rather than an end-to-end optimizer. Based only on local and neighbor historical trajectories, the LLM helps agents adapt both their internal action behaviors and their external cooperation patterns. To make such cognitive intervention stable and resource-aware, we further introduce a phased cognitive guidance strategy that selectively activates different forms of adaptation over the course of optimization. In this way, LAC-MAS provides an initial step from handcrafted distributed coordination toward trajectory-driven self-design in multi-agent black-box optimization.

We summarize our main contributions as follows:
\begin{itemize}
    \item We redesign the agent-level swarm optimization dynamics for distributed black-box consensus optimization, introducing an adaptive internal mechanism that better balances exploration, convergence, and local escape under decentralized consensus constraints.
    
    \item On top of this adaptive substrate, we propose \textbf{Learning to Act and Cooperate (LAC-MAS)}, a trajectory-driven self-design framework in which LLMs provide sparse high-level guidance to jointly shape how agents act locally and cooperate globally from historical optimization trajectories.
    
    \item We develop a phased cognitive guidance strategy that enables stage-aware coordination of trajectory-driven action and cooperation guidance, making high-level adaptive guidance practical for decentralized black-box optimization.
\end{itemize}

\section{Related Work}
\textbf{Distributed Black-Box Consensus Optimization.}
Distributed black-box optimization studies how multiple agents cooperatively optimize a global objective using only local function evaluations and neighbor communication. Representative approaches include decentralized zeroth-order optimization~\cite{pmlr-v202-mhanna23a} and unified perspectives connecting zeroth-order and first-order decentralized optimization under nonconvex and stochastic settings~\cite{pmlr-v235-sahinoglu24a}. Related works have further examined asynchronous decentralized optimization~\cite{pmlr-v202-nabli23a}, time-varying network effects on consensus behavior~\cite{metelev2023consensus,nedic2018network}, and fundamental complexity trade-offs in decentralized training~\cite{pmlr-v139-lu21a}. These studies provide important foundations for distributed optimization under communication constraints, but they largely focus on fixed update rules or predefined optimization mechanisms, rather than learning how agents should adapt their local behaviors and cooperative interactions from historical optimization trajectories.

\textbf{Learning-Driven Optimization Design.}
Recent work has increasingly explored learning-based mechanisms for algorithm design, suggesting that optimization strategies can be adapted from feedback, historical trajectories, and higher-level performance signals. Large language models have been used as optimization engines for iterative refinement~\cite{yang2024opro}, for prompt and instruction optimization~\cite{pryzant2023automatic}, for trajectory-aware online strategy adaptation~\cite{wan2024apo}, and for enhancing Bayesian optimization through LLM-guided reasoning~\cite{liu2024llambo}, and for serving as meta-surrogates in offline data-driven many-task optimization~\cite{Zhang2025meta}. Related efforts also investigate generating optimization algorithms or solver code directly from natural language specifications~\cite{ma2024llamoco}. More broadly, recent LLM-based multi-agent systems demonstrate that LLMs can coordinate structured agent behaviors through role assignment, interaction design, and adaptive collaboration~\cite{pmlr-v235-zhuge24a,wu2024autogen,ye2025masgpt}. However, most existing learning-driven or LLM-assisted methods are developed for centralized or single-agent settings, or focus on general collaborative reasoning rather than distributed black-box consensus optimization with only local and neighbor information.

\textbf{Adaptive Coordination in Decentralized Systems.}
A smaller line of work studies how coordination structures in decentralized optimization can be made adaptive, including learning communication graphs under data heterogeneity~\cite{pmlr-v206-le-bars23a}, analyzing the role of network heterogeneity in decentralized convergence~\cite{Koloskova*2020Decentralized}, and designing structure-aware or communication-efficient mixing strategies~\cite{lian2017decentralized}. These studies highlight the importance of coordination design in distributed optimization. Recent advances also demonstrate the scalability of Large Language Model-based multi-agent collaboration in complex tasks~\cite{qian2025scaling}. However, they typically focus on graph adaptation, participation patterns, or communication efficiency at the network level. In contrast, our work remains in a fixed communication topology and studies how historical optimization trajectories can guide both agent-internal action adaptation and agent-external cooperation adaptation within the given decentralized interaction structure.

\section{Problem Formulation}
We consider a distributed black-box consensus optimization problem over a connected communication graph $\mathcal{G}=(\mathcal{V},\mathcal{E})$, where $\mathcal{V}=\{1,\ldots,N\}$ denotes the set of agents and $\mathcal{E}$ denotes the set of communication links. Each agent $i\in\mathcal{V}$ can communicate only with its neighbors in the fixed graph, denoted by $\mathcal{N}_i=\{k\mid(i,k)\in\mathcal{E}\}$. This setting reflects the common case in distributed black-box multi-agent optimization, where communication links are typically determined by the underlying system architecture, physical connectivity, or communication range, and therefore remain fixed or change only passively in special situations. In this work, we focus on the fixed connected topology setting, while allowing the proposed cooperation mechanism to adapt only the weights assigned to existing neighbors.

Each agent $i$ is associated with a local black-box objective function $f_i:\mathbb{R}^D\to\mathbb{R}$, which can only be queried by function evaluation. The agents jointly aim to optimize the global objective defined by the average of local objectives,
\begin{equation}
f(x)=\frac{1}{N}\sum_{i=1}^{N} f_i(x),
\end{equation}
while satisfying a consensus requirement across agents. Since the local objectives are black-box, gradient information is unavailable, and optimization must rely entirely on function evaluations and decentralized interaction.

Accordingly, the goal is to find agent states $\{x_i\}_{i=1}^N$ such that the system simultaneously achieves: 
(i) \emph{objective improvement}, namely reducing the global objective value $f(x)$; and 
(ii) \emph{consensus}, namely driving all local agent states toward agreement. The corresponding distributed consensus optimization problem can be written as
\begin{equation}
\min_{\{x_i\}_{i=1}^N}\ \frac{1}{N}\sum_{i=1}^{N} f_i(x_i)
\quad \text{s.t.} \quad x_i=x_j,\ \forall i,j.
\end{equation}
Equivalently, the optimization process seeks a consensus solution at which all agents agree on a common decision variable while minimizing the average local objective.

Under this setting, each agent maintains only local optimization information together with messages received from its neighbors. In particular, agent $i$ can access its own historical query results, local particle states, and aggregated trajectory statistics from neighboring agents, but cannot access other agents' objective functions, gradients, or any global optimization state. During optimization, the global objective value is not directly available to any agent and is used only for offline evaluation in experiments. This information pattern defines the decentralized black-box consensus setting studied in this paper.

\section{Methodology}
\subsection{Framework Overview}
\begin{figure*}[!t]
\centering
\centerline{\includegraphics[width=0.95\textwidth]{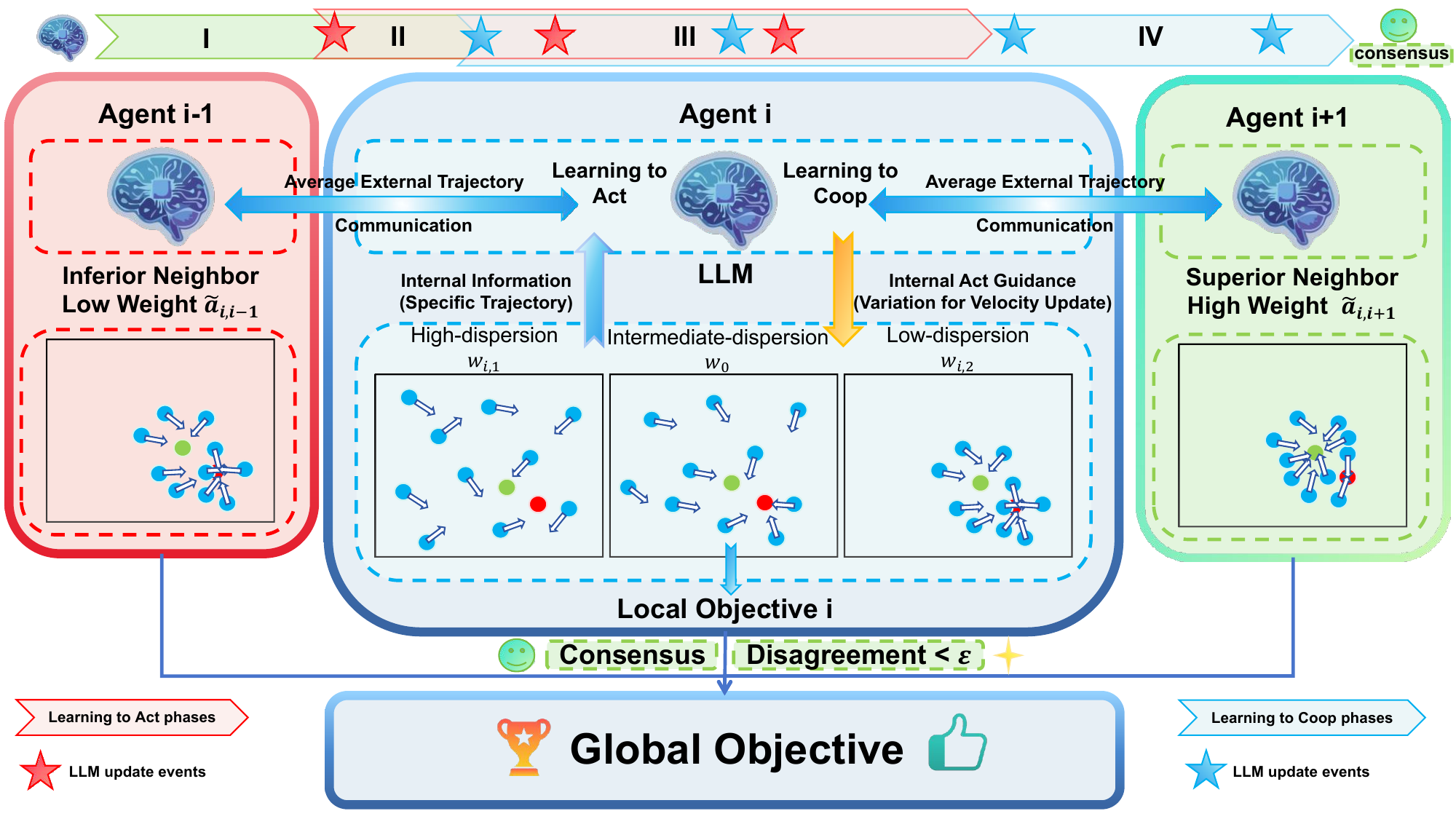}}
\caption{LAC-MAS Framework}
\label{fig:LAC-MAS Framework}
\end{figure*}
As illustrated in Fig.~\ref{fig:LAC-MAS Framework}, LAC-MAS adopts a fully decentralized framework in which each agent consists of two tightly coupled layers: an adaptive swarm-based execution layer and a trajectory-driven guidance layer. The execution layer carries out local black-box optimization under decentralized consensus constraints, while the guidance layer periodically updates how the agent should act internally and how it should cooperate externally based on accumulated optimization trajectories.

At the execution level, each agent maintains a local population of particles as its black-box optimizer. This population-based design allows the agent to explore the search space without gradient information and to encode its current optimization status through collective particle dynamics. Different from directly applying a conventional particle swarm optimizer in a distributed consensus setting, we redesign the agent-level swarm dynamics with an adaptive internal mechanism that better matches the needs of decentralized black-box optimization, particularly in balancing exploration, convergence, and local escape.

On top of this adaptive execution layer, each agent is equipped with an agent-specific large language model that serves as a high-level guidance module rather than an end-to-end optimizer. The LLM does not directly update decision variables. Instead, it operates on structured historical trajectory information, including local population evolution and neighbor-provided trajectory summaries, to produce two forms of guidance: \emph{learning to act}, which refines the agent's internal search behavior, and \emph{learning to cooperate}, which adjusts how the agent weights information from different neighbors during consensus formation.

To coordinate these two forms of guidance, LAC-MAS further introduces a phased cognitive guidance mechanism. Rather than intervening at every iteration, this mechanism organizes when trajectory-driven guidance should be refreshed over the course of optimization, so that internal action adaptation and external cooperation adaptation can be aligned with the evolving optimization regime in a stable and resource-aware manner. Overall, LAC-MAS forms a trajectory-driven collaborative framework in which decentralized black-box optimization is supported by adaptive execution, high-level guidance, and stage-aware coordination.

\subsection{Adaptive Swarm Execution Layer for Learning to Act}

A central challenge in distributed black-box consensus optimization is that each agent must regulate its local search behavior using only black-box feedback and limited neighbor communication. In this setting, the search direction favored by an agent's local objective is not necessarily aligned with the globally desirable consensus solution, since the final consensus point is determined by the collective objective rather than any single local optimum. As a result, directly applying conventional particle swarm dynamics is often insufficient: fixed internal update patterns cannot flexibly balance exploration, convergence, and escape from poor local attractors under decentralized constraints. To address this issue, we endow each agent with an adaptive swarm execution layer that characterizes its current search regime through the dispersion of its local particle population.

\textbf{Adaptive internal action mechanism.}
Specifically, for agent $i$, let $\{x_{i,p}^{(t)}\}_{p=1}^{P}$ denote its local particle population at iteration $t$, where $P$ is the population size and $x_{i,p}^{(t)} \in \mathbb{R}^D$ is the position of particle $p$ in the $D$-dimensional decision space. We define the particle centroid of agent $i$ as
\begin{equation}
\mu_i^{(t)} = \frac{1}{P}\sum_{p=1}^P x_{i,p}^{(t)},
\end{equation}
and the corresponding particle divergence as
\begin{equation}
D_i^{(t)} = \frac{1}{P}\sum_{p=1}^P \|x_{i,p}^{(t)} - \mu_i^{(t)}\|_2^2,
\end{equation}
where $D_i^{(t)}$ measures the dispersion of the local particle population around its centroid. A larger divergence indicates that the particle population is more widely spread and the agent remains in a more exploratory search regime, whereas a smaller divergence reflects stronger concentration and a tendency toward local convergence.

Based on this population-state signal, we introduce for each agent a small set of internal behavioral coefficients,
\begin{equation}
\mathbf{w}_i = (w_{i,1},\, w_0,\, w_{i,2}),
\end{equation}
where $w_{i,1}$, $w_0$, and $w_{i,2}$ correspond to different internal search regimes. During execution, agent $i$ selects an active coefficient $w_i^{(t)}$ according to the current divergence level:
\begin{equation}
w_i^{(t)} =
\begin{cases}
w_{i,2}, & D_i^{(t)} < d_1,\\
w_0, & d_1 \le D_i^{(t)} \le d_2,\\
w_{i,1}, & D_i^{(t)} > d_2,
\end{cases}
\end{equation}
where $d_1$ and $d_2$ are divergence thresholds satisfying $d_1 < d_2$. The selected coefficient $w_i^{(t)} \in \mathbb{R}$ is then used to regulate the internal swarm dynamics. Let $v_{i,p}^{(t)} \in \mathbb{R}^D$ denote the velocity of particle $p$ at iteration $t$, and let $\Delta_{i,p}^{(t)} \in \mathbb{R}^D$ denote the random modulation vector generated by the underlying swarm update rule. Using element-wise multiplication $\odot$, the particle velocity is updated as
\begin{equation}
v_{i,p}^{(t+1)} = w_i^{(t)} \big( \Delta_{i,p}^{(t)} \odot v_{i,p}^{(t)} \big),
\end{equation}
followed by the position update
\begin{equation}
x_{i,p}^{(t+1)} = x_{i,p}^{(t)} + v_{i,p}^{(t+1)}.
\end{equation}
This multiplicative form preserves the stochastic modulation mechanism of the underlying swarm dynamics, while introducing the adaptive coefficient $w_i^{(t)}$ to regulate the strength of internal variation across different search regimes. Consequently, the execution layer itself becomes adaptive to the evolving local optimization state, providing an internal optimization substrate better suited to decentralized black-box consensus optimization.

This adaptive mechanism already improves the compatibility of swarm execution with distributed black-box consensus optimization, because it allows the internal search behavior of each agent to vary with its current population state rather than remain fixed throughout optimization. However, if the coefficient set $\mathbf{w}_i$ is manually specified once and kept unchanged, the resulting adaptation is still rule-based and cannot fully exploit richer trajectory-level evidence accumulated over time.

\textbf{Trajectory-driven refinement by LLM guidance.}
Built on top of the above adaptive execution layer, the LLM further refines internal action behavior using historical optimization trajectories. Rather than directly controlling particle updates, the LLM infers the internal coefficient set $\mathbf{w}_i$ from recent trajectory information, so that the resulting search regimes are not fixed \emph{a priori} but refreshed according to the agent's evolving optimization history. The divergence-based rule in Eq.~(4) then instantiates these learned coefficients online, enabling the agent to translate trajectory-level guidance into population-state-dependent execution. Therefore, learning to act in LAC-MAS combines a population-state-driven adaptive execution mechanism with higher-level trajectory-driven refinement, enabling more flexible internal behavior adaptation than handcrafted static swarm rules.

\subsection{Trajectory-Driven Cooperation Guidance}

\textbf{Learning agent-external cooperative behaviors.}
Beyond regulating how each agent acts internally, LAC-MAS also learns how agents should cooperate during consensus formation. In distributed black-box consensus optimization, information received from different neighbors may have highly unequal utility: some neighbors may provide more informative optimization trajectories, while others may be less helpful due to slow progress, premature concentration, or locally biased search behavior. Therefore, relying on fixed or uniform neighbor weighting can limit the flexibility and efficiency of decentralized coordination.

To address this issue, agent $i$ evaluates each neighbor $k \in \mathcal{N}_i$ through a trajectory-based descriptor
\begin{equation}
\mathbf{s}_{ik}^{(t)} =
\big[
\overline{f}_k^{(t)},\;
\overline{D}_k^{(t)},\;
\overline{\|\Delta x_k\|}^{(t)}
\big],
\end{equation}
where $\overline{f}_k^{(t)}$ denotes the recent average objective value of neighbor $k$ over a short trajectory window, $\overline{D}_k^{(t)}$ denotes its recent average particle divergence, and $\overline{\|\Delta x_k\|}^{(t)}$ denotes the recent average magnitude of its state variation. Together, these quantities summarize the neighbor's solution quality, population dispersion, and recent search activity.

Based on the descriptor set $\{\mathbf{s}_{ik}^{(t)}\}_{k\in\mathcal{N}_i}$, the LLM outputs candidate cooperation weights over $\mathcal{N}_i \cup \{i\}$. To guard against occasional infeasible outputs, we apply an explicit normalization/projection step before execution, ensuring that the final weights $\{a_{ik}^{(t)}\}$ are nonnegative and sum to one. The resulting weights are then used in the agent-external cooperative update
\begin{equation}
x_i^{(t+1)} = \sum_{k \in \mathcal{N}_i \cup \{i\}} a_{ik}^{(t)}\, x_k^{(t+1)}.
\end{equation}

In this way, cooperation in LAC-MAS is no longer governed by fixed neighbor averaging, but by trajectory-driven learning of the relative value of neighbor information within the given communication graph. Importantly, the topology itself remains unchanged: the learned cooperation mechanism only adjusts the influence of existing neighbors based on their historical optimization utility, thereby preserving decentralized communication constraints while improving coordination flexibility.

\begin{algorithm}[t]
\caption{Trajectory-Driven Internal Action and Cooperation Update (Agent $i$)}
\label{alg:lacmas_agent_update}
\begin{algorithmic}
\STATE {\bf Input:} particle states $\{x_{i,p}^{(t)}, v_{i,p}^{(t)}\}_{p=1}^P$, neighbor set $\mathcal{N}_i$, local history $\mathcal{H}_i^{(t)}$, neighbor histories $\{\mathcal{H}_k^{(t)}\}_{k\in\mathcal{N}_i}$
\STATE Compute centroid $\mu_i^{(t)}$ and divergence $D_i^{(t)}$
\STATE Aggregate local trajectory features from $\mathcal{H}_i^{(t)}$
\STATE Construct neighbor descriptors $\{\mathbf{s}_{ik}^{(t)}\}_{k\in\mathcal{N}_i}$
\IF{guidance refresh is triggered}
    \STATE Build an agent-specific prompt using local trajectory features and neighbor descriptors
    \STATE Query the LLM to obtain:
    \STATE \hspace{1em} internal coefficient set $\mathbf{w}_i = (w_{i,1}, w_0, w_{i,2})$
    \STATE \hspace{1em} candidate cooperation weights over $\mathcal{N}_i \cup \{i\}$
    \STATE Project outputs to feasible ranges
\ENDIF
\STATE Select active internal coefficient $w_i^{(t)}$ based on $D_i^{(t)}$
\STATE Update particle velocities and positions using adaptive internal action
\STATE Normalize/project candidate cooperation weights to obtain $\{a_{ik}^{(t)}\}$
\STATE Perform weighted consensus update
\STATE Update local history $\mathcal{H}_i^{(t+1)}$
\end{algorithmic}
\end{algorithm}

\subsection{Phased Cognitive Guidance}
While Secs.~4.2--4.3 specify how LAC-MAS learns agent-internal action guidance and agent-external cooperation guidance from historical trajectories, an additional question remains: \emph{when} should such guidance be refreshed during optimization? In distributed black-box consensus optimization, these two forms of guidance have different functional roles and therefore different refresh requirements. Internal action guidance mainly matters when the local search regime changes substantially, whereas cooperation guidance needs to be revisited more regularly as the relative utility of neighbors evolves during consensus formation. Moreover, from an execution perspective, the low-level swarm optimizer evolves continuously and can run in parallel across agents, whereas high-level LLM guidance is refreshed only intermittently and may be obtained asynchronously. Applying a uniform guidance-refresh pattern throughout optimization is therefore both inefficient and potentially unstable.

To address this issue, we introduce \emph{Phased Cognitive Guidance} (PCG), a high-level scheduling mechanism that determines when trajectory-driven guidance should be refreshed during optimization. PCG does not alter the underlying learning mechanisms themselves. Instead, it coordinates the refresh of internal action guidance and external cooperation guidance according to their distinct functional roles, while decoupling continuous decentralized optimization from sparse high-level guidance updates without requiring strict iteration-level synchronization.

\textbf{Pre-experiment calibration.}
PCG uses a lightweight pre-experiment to estimate a characteristic optimization horizon $T$, which serves as a coarse temporal reference for scheduling guidance refresh. The purpose of $T$ is not to predict the exact convergence time, but to provide a stable scale for organizing stage-aware updates. When preliminary calibration is unavailable, a coarse estimate of $T$ is still sufficient for applying the framework.

\textbf{Guidance-refresh gating.}
Based on the calibrated horizon $T$, PCG defines two binary gating functions indicating whether cooperation guidance or internal action guidance is refreshed at iteration $t$, respectively. The cooperation-refresh gate is defined as
\begin{equation}
g_{\mathrm{ext}}(t)=\mathbb{I}\{t\in\mathcal{T}_{\mathrm{ext}}\},
\end{equation}
where
\begin{equation}
\mathcal{T}_{\mathrm{ext}}=\{\lceil m\rho_{\mathrm{ext}}T\rceil\}_{m\ge 1},
\end{equation}
and $\rho_{\mathrm{ext}}>0$ controls the refresh interval of cooperation guidance.

The internal-action refresh gate is defined as
\begin{equation}
g_{\mathrm{int}}(t)=\mathbb{I}\{t\in\mathcal{T}_{\mathrm{int}}\},
\end{equation}
where
\begin{equation}
\mathcal{T}_{\mathrm{int}}=\{\lceil \rho_1T\rceil,\;\lceil \rho_2T\rceil\},
\qquad 0<\rho_1<\rho_2<1.
\end{equation}
Here, $\rho_1$ and $\rho_2$ specify two key refresh points for internal action guidance. This asymmetric design reflects the distinct roles of the two adaptation types: cooperation guidance needs to be revisited repeatedly as neighbor utilities evolve, whereas internal action guidance mainly needs to be refreshed when the local search regime changes substantially. After the optimization progresses beyond the calibrated horizon, internal-action refresh is deactivated:
\begin{equation}
g_{\mathrm{int}}(t)=0,\qquad \forall t\ge T.
\end{equation}

The interaction between the two gates induces an implicit stage structure over the optimization process. Let $s(t)$ denote the stage index at iteration $t$, defined by transition points $\tau_m=\lceil \alpha_m T\rceil$ with $0<\alpha_1<\alpha_2<\alpha_3\le 1$:
\begin{equation}
s(t)=
\begin{cases}
1, & 0\le t<\tau_1,\\
2, & \tau_1\le t<\tau_2,\\
3, & \tau_2\le t<\tau_3,\\
4, & \tau_3\le t.
\end{cases}
\end{equation}
This stage structure summarizes how the emphasis of guidance refresh shifts over time, including initial trajectory accumulation, adaptive internal search regulation, joint action--cooperation refinement, and late-stage consensus stabilization. The corresponding qualitative descriptions are provided in Appendix~A.

Accordingly, PCG should be viewed primarily as a guidance-refresh scheduling mechanism, while the stage structure serves as a coarse interpretation of how refresh emphasis evolves across optimization. By coordinating the refresh of internal action guidance and external cooperation guidance in this manner, PCG enables stable and resource-aware learning to act and learning to cooperate throughout distributed black-box optimization.

\subsection{Consensus Guarantees}
We analyze the consensus properties of LAC-MAS and show that the proposed internal action adaptation, trajectory-driven cooperation guidance, and phased cognitive refresh preserve the consensus structure of the underlying decentralized swarm optimizer under standard assumptions. Our goal here is not to model the full stochastic token-generation process of the LLM, but to establish that the closed-loop optimization dynamics remain admissible for consensus.

\textbf{Assumptions.}
We adopt the following assumptions, which are standard in consensus-based distributed optimization and consistent with the deterministic consensus backbone established for MASOIE~\cite{chen2025masoie}.
\begin{itemize}
    \item[(A1)] The communication graph $\mathcal{G}=(\mathcal{V},\mathcal{E})$ is fixed and connected.
    \item[(A2)] After normalization/projection, the cooperation weights $\{a_{ik}^{(t)}\}$ are nonnegative, graph-compatible, and row-stochastic for all $t$, i.e.,
    \begin{equation}
    \begin{aligned}
    &a_{ik}^{(t)} \ge 0,\\
    &a_{ik}^{(t)} = 0 \quad \text{if } k \notin \mathcal{N}_i \cup \{i\},\\
    &\sum_{k \in \mathcal{N}_i \cup \{i\}} a_{ik}^{(t)} = 1.
    \end{aligned}
    \end{equation}
    \item[(A3)] The internal coefficients selected from $\mathbf{w}_i=(w_{i,1},w_0,w_{i,2})$ and the modulation vectors $\Delta_{i,p}^{(t)}$ are bounded.
    \item[(A4)] Under fixed internal coefficients and admissible consensus weights, the underlying decentralized swarm executor admits consensus, consistent with the deterministic MASOIE backbone~\cite{chen2025masoie}.
    \item[(A5)] In the late-stage stable regime induced by PCG, the effective perturbation term entering the agent-level consensus fusion is asymptotically vanishing, i.e., $\|\xi^{(t)}\|\to 0$ as $t\to\infty$.
\end{itemize}

\textbf{Admissibility of trajectory-driven cooperation.}
The cooperation mechanism in Sec.~4.3 does not modify the communication topology itself; it only reweights existing neighbor information. Because the final cooperation weights are explicitly normalized/projected before execution, the induced matrix
\begin{equation}
A^{(t)}=[a_{ik}^{(t)}]
\end{equation}
remains nonnegative, row-stochastic, and graph-compatible for all $t$. Hence, the LLM-guided cooperation module preserves the standard consensus mixing conditions over the original connected graph.

\textbf{Boundedness of internal action adaptation.}
The internal action mechanism in Sec.~4.2 modulates the swarm dynamics through bounded coefficients selected from the finite set $\mathbf{w}_i$. Since both $w_i^{(t)}$ and $\Delta_{i,p}^{(t)}$ are bounded, the internal update remains a bounded modulation of the underlying swarm executor. Moreover, by PCG, internal-guidance refresh occurs only at finitely many scheduled stages and is deactivated after the calibrated horizon $T$. Therefore, the internal action mechanism introduces only piecewise-constant, finite-stage adaptation rather than persistent high-frequency switching.

\textbf{Closed-loop consensus dynamics.}
Let $x_i^{(t)}$ denote the agent-level representative state of agent $i$ used in consensus fusion at iteration $t$, as induced by its local particle population. Stacking these representative states across agents gives
\begin{equation}
x^{(t)} = [x_1^{(t)}, \ldots, x_N^{(t)}]^\top .
\end{equation}
The cooperative update of LAC-MAS can be written as
\begin{equation}
x^{(t+1)} = A^{(t)}x^{(t)} + \xi^{(t)},
\end{equation}
where $A^{(t)}$ is the row-stochastic mixing matrix induced by the cooperation weights and $\xi^{(t)}$ collects the perturbations introduced by local black-box swarm search and internal adaptive execution. By Assumptions (A2)--(A3), $A^{(t)}$ remains admissible and $\xi^{(t)}$ remains bounded; by Assumption (A5), the perturbation vanishes asymptotically.

Under PCG, the resulting system is a switched consensus process with admissible time-varying mixing matrices and finitely many internal-guidance regime changes. After the final internal refresh, the system evolves under stable decentralized dynamics with vanishing perturbations.

\begin{theorem}
Under Assumptions (A1)--(A5), the proposed LAC-MAS framework preserves consensus, i.e.,
\begin{equation}
\lim_{t\to\infty}\|x_i^{(t)}-x_j^{(t)}\| = 0,\qquad \forall i,j.
\end{equation}
\end{theorem}
\textbf{Proof sketch.}
For fixed internal coefficients and fixed admissible consensus weights, the claim follows from the consensus backbone of the underlying decentralized swarm optimizer, consistent with the deterministic MASOIE analysis. The LLM-guided cooperation mechanism preserves graph connectivity and row-stochasticity through explicit normalization/projection, so it does not alter the admissible consensus structure. The internal action mechanism introduces only bounded coefficient modulation, and PCG restricts internal-guidance refresh to finitely many scheduled stages before $T$. Consequently, LAC-MAS can be interpreted as a switched consensus system with admissible time-varying mixing matrices and asymptotically vanishing perturbations. Standard consensus arguments for connected row-stochastic switching systems then imply
\(
\|x_i^{(t)}-x_j^{(t)}\| \to 0
\)
for all agents. Detailed verification that LAC-MAS satisfies the admissibility, finite-switching, and vanishing-perturbation conditions required by this argument is provided in Appendix~E. \hfill$\square$

This result shows that LAC-MAS extends the decentralized consensus backbone with trajectory-driven internal and external adaptation, while retaining the structural conditions required for consensus.

\section{Experimental Setup and Results}
\subsection{Experimental Setup}
\textbf{Benchmark Functions.}
We evaluate LAC-MAS on a standard benchmark suite for consensus-based distributed black-box optimization, consisting of ten test functions (F1--F10) widely adopted in prior multi-agent optimization studies~\cite{chen2025masoie}. The benchmarks cover diverse characteristics, including unimodal and multimodal landscapes, homogeneous and heterogeneous objective distributions, as well as shifted and non-separable functions.

All benchmark problems are instantiated with 100 decision variables and distributed across 20 agents. Each agent can only query its own local black-box objective, while the global objective---defined as the average of all local objectives---is never accessible during optimization. This strictly follows the consensus-based distributed optimization protocol~\cite{chen2025masoie}.

\textbf{Compared Methods.}
We compare LAC-MAS with representative state-of-the-art baselines from three categories:
(1) multi-agent swarm optimization methods, including MASOIE~\cite{chen2025masoie}, which serves as the primary baseline due to its adaptive internal--external learning design;
(2) consensus-driven population-based frameworks, including GFPDO~\cite{ai2017general}, which employ explicit consensus mechanisms and typically incur relatively dense communication overhead;
and (3) classical distributed gradient-free and swarm-based methods, including RGF~\cite{yuan2014randomized} and DA-PSO~\cite{jalloul2015distributed}, which represent representative average-consensus-based black-box coordination strategies.

All baselines adopt the parameter settings reported in their original publications to ensure fair comparison. Each algorithm is independently executed 25 times, consistent with prior benchmark protocols~\cite{chen2025masoie}. Convergence is declared when the disagreement metric drops below $10^{-7}$.

\textbf{Evaluation Metrics.}
Performance is evaluated using three metrics: (i) final fitness, (ii) cumulative communication cost until convergence, and (iii) disagreement over iterations. Statistical comparisons across benchmark functions are conducted using the Friedman test followed by the Nemenyi post-hoc procedure at significance level $\alpha=0.05$.

\textbf{Implementation and Runtime Environment.}
Unless otherwise stated, all experiments are conducted in a distributed simulation setting on a computing server equipped with Intel(R) Xeon(R) CPU E5-2699 v3 @ 2.30GHz and CentOS 7.5 x64, following the same hardware setting reported for MASOIE-based evaluation~\cite{chen2025masoie}. The low-level swarm executors run continuously for each agent, while high-level LLM guidance is invoked sparsely according to PCG rather than at every iteration.

\subsection{Benchmark Results}
Quantitative results on benchmark functions are summarized in Table 1, while representative convergence curves of disagreement are illustrated in Fig. 2.
\begin{table*}[htbp]
  \centering
  \caption{COMPARISON OF LAC-MAS WITH EXISTING ALGORITHMS}
  \label{tab:exact_match}
  \adjustbox{max width=\linewidth}{
\begin{tabular}{ccccccccccccc}
\toprule[1pt] 
    \midrule[0.5pt]
   &   & LAC-MAS & MASOIE & GFPDO & RGF & DAPSO &  & LAC-MAS & MASOIE & GFPDO & RGF & DAPSO \\     
    \midrule[0.5pt]
mean &   & {\color[HTML]{FF0000} 2.21E+04} & 6.81E+04 & 2.44e+06 & 4.58e+08 & 1.10e+07 &
   & 6.43E+06 & {\color[HTML]{FF0000} 6.03E+06} & 2.77e+09 & 2.98e+09 & 1.82e+09 \\
median &  & 2.14E+04 & 6.61E+04 & 2.34e+06 & 4.50e+08 & 9.52e+06 &
   & 6.62E+06 & 5.98E+06 & 2.84e+09 & 2.99e+09 & 1.80e+09 \\
std &   & 4.65E+03 & 1.36E+04 & 6.97e+05 & 1.96e+07 & 4.39e+06 &
   & 1.16E+06 & 1.96E+06 & 2.17e+08 & 2.71e+07 & 3.18e+08 \\
p-value & \multirow{-4}{*}{F1} & - & 2.53e-02* & 7.74e-06\# & 2.22e-16\# & 1.97e-11\# &
\multirow{-4}{*}{F6}  & - & 7.88E-01 & 1.04e-08\# & 2.23e-13\# & 9.35e-04\# \\
    \midrule[0.5pt]
mean &   & {\color[HTML]{FF0000} 7.57E+04} & 7.96E+04 & 1.57e+07 & 1.45e+08 & 7.02e+07   &
& {\color[HTML]{FF0000} 1.06E+05} & 1.47E+05 & 2.31e+09 & 2.82e+09 & 2.05e+09 \\
median &   & 7.58E+04 & 7.93E+04 & 1.46e+07 & 1.45e+08 & 6.94e+07 &
& 1.03E+05 & 1.47E+05 & 2.32e+09 & 2.82e+09 & 1.89e+09 \\
std &   & 5.56E+03 & 9.61E+03 & 7.22e+06 & 0.00e+00 & 8.22e+06 &
& 2.59E+04 & 4.32E+04 & 1.75e+08 & 0.00e+00 & 4.62e+08 \\
p-value & \multirow{-4}{*}{F2} & - & 4.21E-01 & 1.72e-04\# & 2.22e-16\# & 2.06e-09\# &
\multirow{-4}{*}{F7}  & - & 1.80E-01 & 6.11e-09\# & 2.22e-16\# & 2.13e-06\# \\
    \midrule[0.5pt]
mean &   & {\color[HTML]{FF0000} 9.99E+03} & 1.05E+04 & 2.60e+09 & 3.57e+07 & 4.02e+08 &
   & {\color[HTML]{FF0000} 3.03E+07} & 4.13E+07 & 4.41e+08 & 1.52e+09 & 1.51e+09 \\
median &   & 9.75E+03 & 1.04E+04 & 2.53e+09 & 3.54e+07 & 4.06e+08 &
   & 3.00E+07 & 4.13E+07 & 4.42e+08 & 1.54e+09 & 1.42e+09 \\
std &   & 1.53E+03 & 3.12E+03 & 2.23e+08 & 7.68e+05 & 7.01e+07 &
   & 3.54E+06 & 3.21E+06 & 2.09e+07 & 2.55e+08 & 5.19e+08 \\
p-value & \multirow{-4}{*}{F3} & - & 7.88E-01 & 1.78e-15\# & 4.86e-04\# & 1.04e-08\# &
\multirow{-4}{*}{F8}  & - & 2.53e-02* & 7.74e-10\# & 1.78e-15\# & 1.44e-14\# \\
    \midrule[0.5pt]
mean &   & 2.36E+01 & 2.42e+01 & {\color[HTML]{FF0000} 1.56e+01} & 2.80e+01 & 3.46e+01 &
   & {\color[HTML]{FF0000} 1.26E+05} & 1.28E+05 & 3.64e+07 & 4.46e+07 & 3.31e+08 \\
median &   & 2.34E+01 & 2.42E+01 & 1.56e+01 & 2.79e+01 & 3.52e+01 &
   & 1.29E+05 & 1.30E+05 & 3.50e+07 & 2.71e+07 & 1.93e+08 \\
std &   & 7.27E-01 & 9.93E-01 & 3.49e-01 & 1.88e+00 & 3.15e+00 &
   & 1.10E+04 & 9.51E+03 & 1.20e+07 & 5.71e+07 & 3.70e+08 \\
p-value & \multirow{-4}{*}{F4} & - & 2.45E-01 & 5.56e-04\# & 8.30e-05\# & 6.76e-10\# &
\multirow{-4}{*}{F9}  & - & 4.21E-01 & 2.93e-08\# & 5.48e-07\# & 1.14e-13\# \\
    \midrule[0.5pt]
mean &   & {\color[HTML]{FF0000} 1.55E+07} & 2.14E+07 & 5.39e+08 & 1.81e+09 & 6.52e+08 &
   & {\color[HTML]{FF0000} 5.06E+06} & 8.28E+06 & 2.30e+09 & 2.46e+09 & 5.52e+09 \\
median &   & 1.61E+07 & 2.16E+07 & 5.32e+08 & 1.81e+09 & 6.73e+08 &
   & 4.87E+06 & 8.07E+06 & 2.30e+09 & 2.42e+09 & 5.26e+09 \\
std &   & 1.87E+06 & 2.04E+06 & 3.78e+07 & 9.19e+06 & 1.07e+08 &
   & 1.17E+06 & 2.00E+06 & 9.69e+07 & 5.15e+08 & 9.95e+08 \\
p-value & \multirow{-4}{*}{F5} & - & 2.53e-02* & 2.13e-06\# & 2.22e-16\# & 1.20e-10\# &
\multirow{-4}{*}{F10} & - & 6.03E-02 & 2.13e-07\# & 1.75e-08\# & 2.22e-16\# \\
    \midrule[0.5pt]
    \bottomrule[1pt] 
\end{tabular}
 }
\vspace{2mm}
\begin{minipage}{1\linewidth}
\footnotesize
* and \# indicate statistical significance based on the Friedman test with Nemenyi post-hoc test at $\alpha=0.05$ and $\alpha=0.01$, respectively.
\end{minipage}
\end{table*}

Across the benchmark suite, LAC-MAS consistently outperforms or matches strong distributed black-box baselines, achieving lower mean and median fitness on most functions while maintaining stable disagreement reduction. The gains are particularly pronounced on functions that require flexible regulation of exploration and convergence, where trajectory-driven internal behavior learning and adaptive cooperation can better balance local refinement and global coordination.

On functions dominated by narrow valleys or strongly directional landscapes (e.g., F3 and F6), LAC-MAS does not exhibit statistically significant differences compared to MASOIE, yet consistently matches the strongest baselines. This behavior suggests that when the optimization structure already favors highly specialized search dynamics, the proposed learning mechanisms preserve robustness without degrading performance, while offering limited room for further improvement.

\subsection{Ablation Experiments}
To assess the individual and joint contributions of learning to act and learning to cooperate, we conduct an ablation study based on MASOIE~\cite{chen2025masoie}. The study progressively enables agent-internal behavior learning and agent-external cooperation learning along two orthogonal dimensions, directly reflecting the core design philosophy of LAC-MAS.

We consider the following four variants:

\textbf{(1) MASOIE (Baseline):}  
A rule-based optimizer with fixed velocity-related coefficients and uniform neighbor weights, serving as a reference without learning capability.

\textbf{(2) LAC-MAS-Coop:}  
Only agent-external cooperation is learned via LLM-guided neighbor weight adaptation, while agent-internal search dynamics remain fixed.

\textbf{(3) LAC-MAS-Act:}  
Only agent-internal behaviors are learned through LLM-guided adaptation of velocity-related coefficients, while agent-external coordination relies on static neighbor weights.

\textbf{(4) LAC-MAS (Full):}  
The complete framework that jointly learns how agents act and how they cooperate from historical optimization trajectories, coordinated through phased cognitive guidance.
\begin{figure*}[t]
  \centering
  \subfigure[F1]{
    \includegraphics[width=0.45\linewidth]{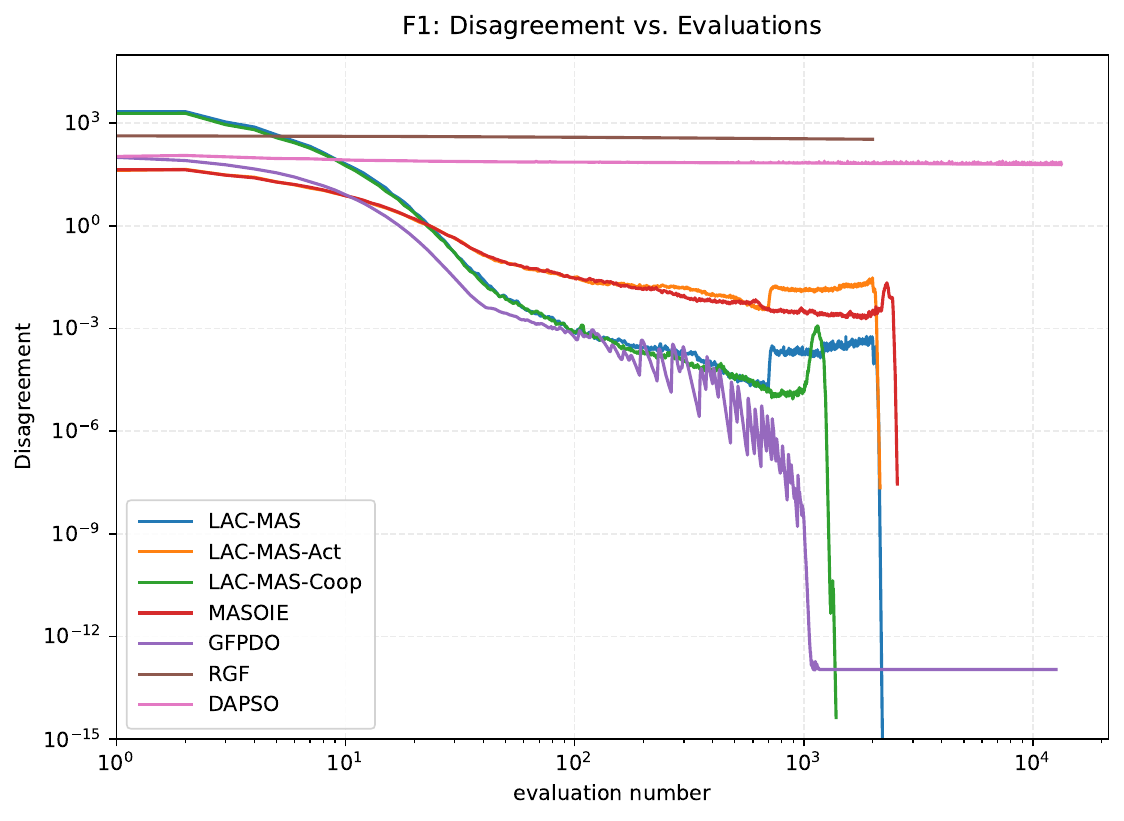}
  }
  \subfigure[F2]{
    \includegraphics[width=0.45\linewidth]{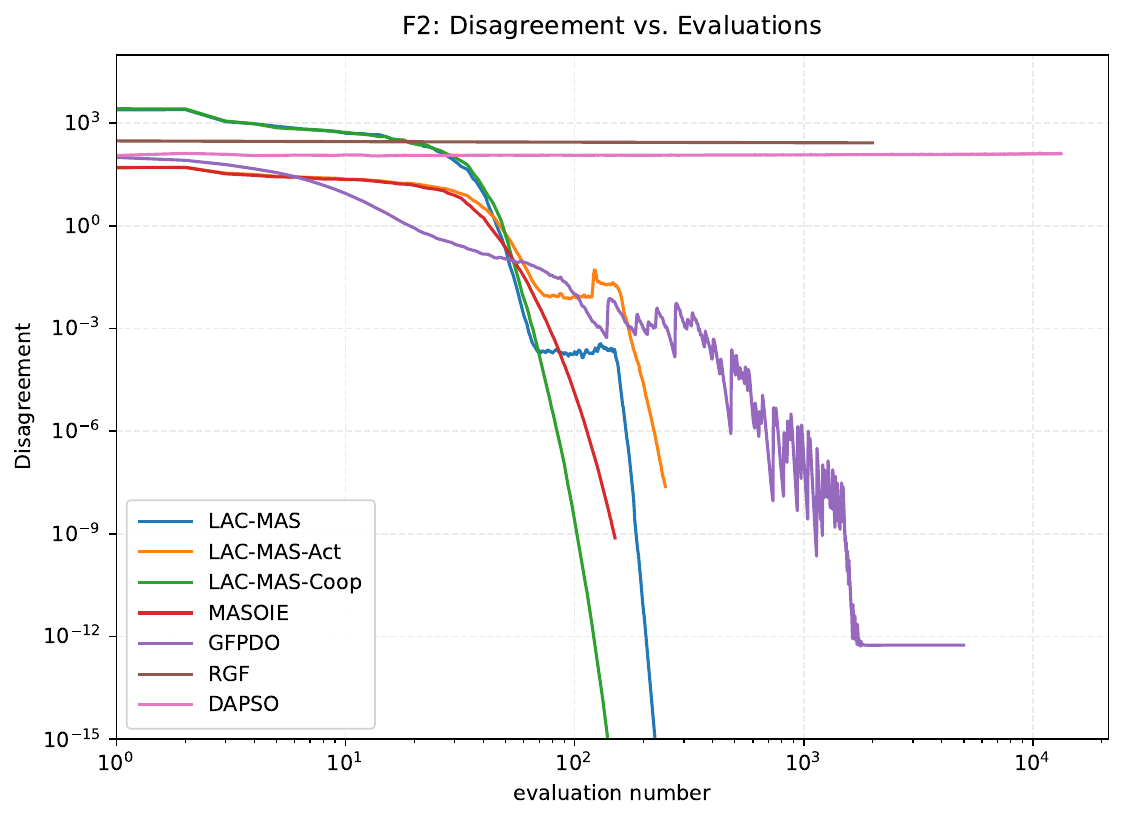}
  }
   \vspace{0.2em}
  \subfigure[F3]{
    \includegraphics[width=0.45\linewidth]{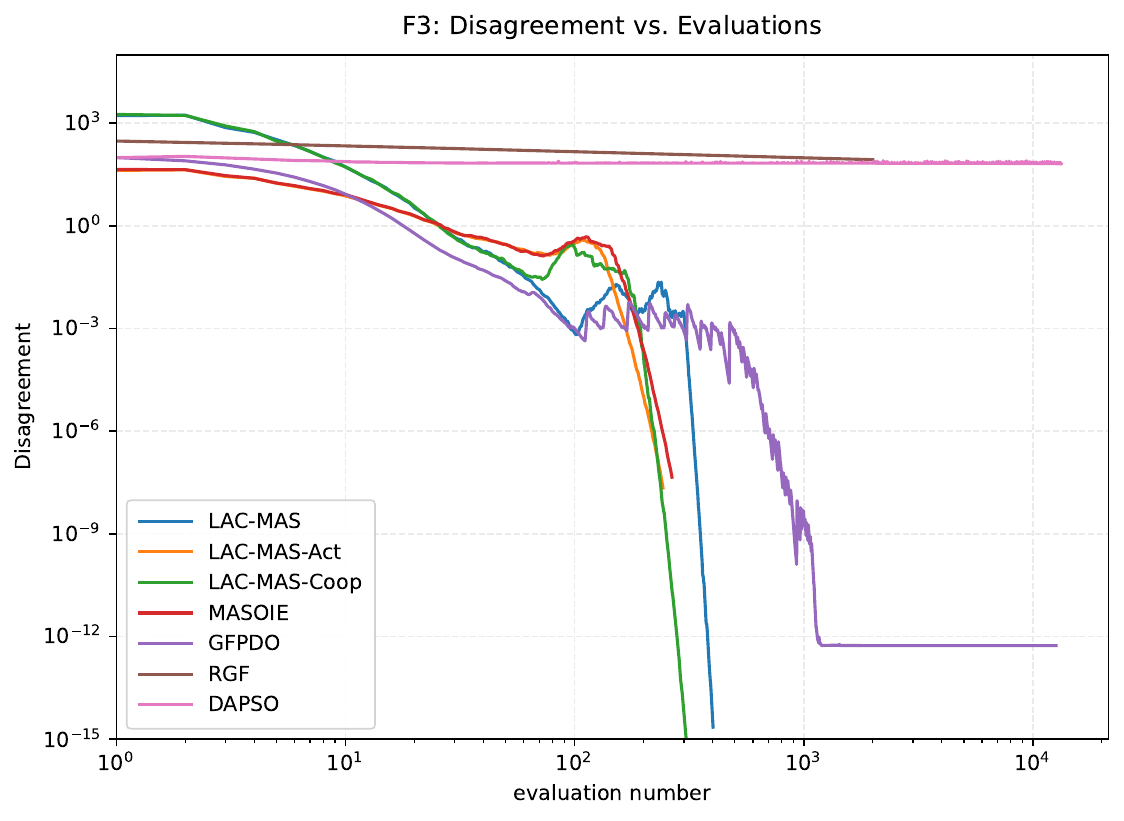}
  }
  \subfigure[F4]{
    \includegraphics[width=0.45\linewidth]{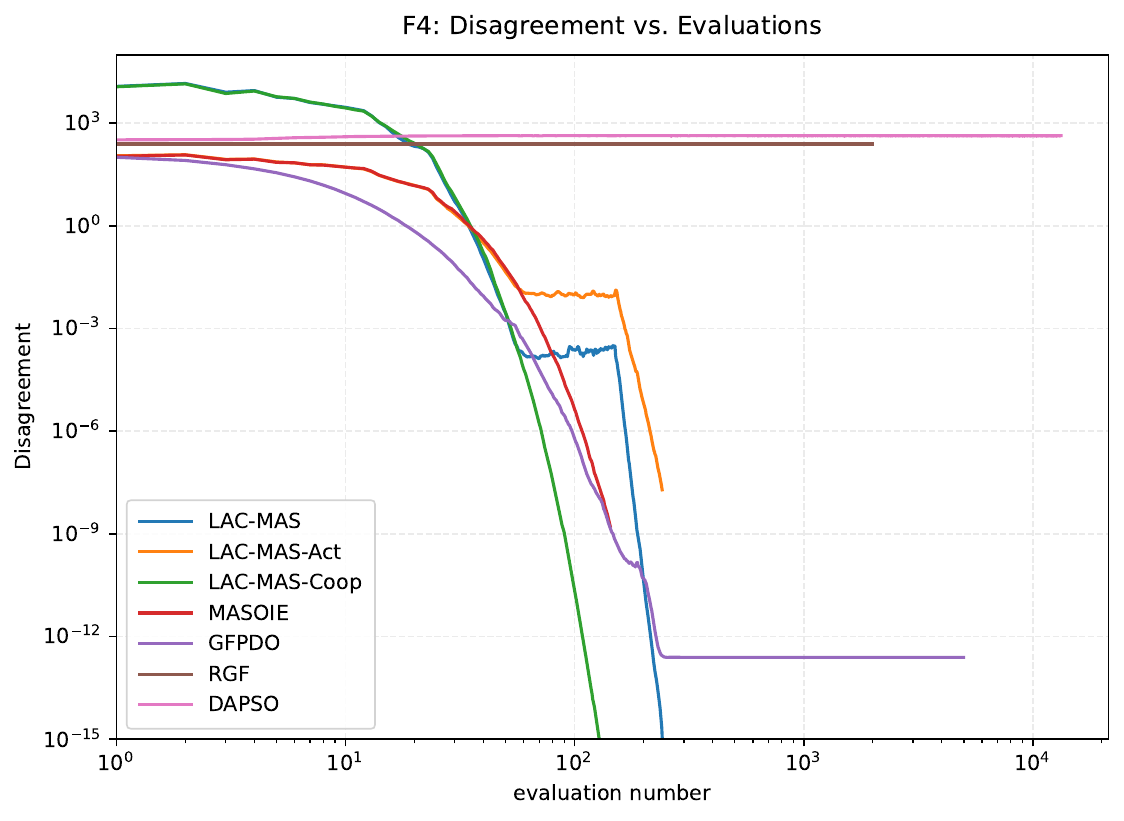}
  }
  \caption{
  Disagreement evolution on representative benchmark functions.}
  \label{fig:disagreement_all}
\end{figure*}
The ablation results reveal clear and complementary roles of learning to act and learning to cooperate.

The Act Learning variant consistently improves final fitness over the baseline, particularly on multimodal and heterogeneous functions. As reflected in the convergence curves, this variant often exhibits rapid objective reduction in early and mid stages, occasionally followed by temporary slowdowns or mild reversals before further improvement. This characteristic behavior indicates that adaptive regulation of agent-internal search dynamics accelerates local refinement while preserving the ability to escape suboptimal regions. At the same time, stronger exploration-induced perturbations may delay or slightly destabilize consensus formation.

In contrast, the Coop Learning variant achieves faster disagreement reduction and lower communication cost across most benchmarks. This observation suggests that learning agent-external cooperation primarily enhances information utilization efficiency and accelerates consensus formation, while offering comparatively limited gains in final objective accuracy when internal search dynamics remain fixed.

The full LAC-MAS framework achieves the most stable and consistently strong performance across benchmarks. By jointly learning agent-internal behaviors and agent-external cooperation and coordinating their activation through phased cognitive guidance, LAC-MAS effectively balances exploration, convergence, and communication efficiency.

Overall, these results demonstrate that learning to act and learning to cooperate address distinct yet interdependent challenges in distributed black-box optimization, and that their coordinated integration is essential for robust and scalable consensus formation. For completeness, detailed fitness statistics of all ablation variants across benchmark functions, along with the remaining disagreement curves, are reported in Appendix~B.

\subsection{Transfer Validation on a Distributed WSN Localization Task}
To further assess the generality of LAC-MAS beyond synthetic benchmarks, we evaluate it on a cooperative multi-target localization task in wireless sensor networks (WSNs), which serves as a representative distributed black-box consensus problem. This task involves partial information, heterogeneous local observations, and limited communication, and therefore provides a meaningful transfer-style validation of the proposed trajectory-driven optimization framework.

Consider $n$ sensors deployed at known locations $\{y_i\}_{i=1}^n$ and $N_t$ targets with unknown 3D positions $\{p_t\}_{t=1}^{N_t}$, where $p_t\in\mathbb{R}^3$. The optimization variable is the concatenation $(p_1,\ldots,p_{N_t})$. The global objective is defined as the average of local objectives across sensors:
\begin{equation}
F(p_1,\ldots,p_{N_t})=\frac{1}{n}\sum_{i=1}^n f_i(p_1,\ldots,p_{N_t}),
\label{eq:wsloc_global}
\end{equation}
where the local objective of sensor $i$ measures the squared mismatch between its received signal measurements and a log-distance path-loss model:
\begin{equation}
f_i(\cdot)=\sum_{t=1}^{N_t}\Big[\phi_{it}-\big(P_0-10n_p\lg\frac{\|p_t-y_i\|}{d_0}\big)\Big]^2 .
\label{eq:wsloc_local}
\end{equation}
Here $\phi_{it}$ denotes the received signal strength (RSS) measurement from target $t$ to sensor $i$, $P_0$ is the reference RSS at distance $d_0$, and $n_p$ is the path-loss exponent. This formulation follows the distributed multi-target localization setting commonly used in prior MASOIE-style evaluation, in which each agent only observes its own local measurement-driven objective and cooperation is required to reduce the system-level estimation error~\cite{chen2025masoie}.

We use the estimation error as the primary metric, defined exactly as the global objective value in Eq.~\eqref{eq:wsloc_global}. For evaluation, at each communication round $k$ we compute a system-level estimate
\begin{equation}
\bar{x}^{(k)}=\frac{1}{n}\sum_{i=1}^n x_i^{(k)},
\end{equation}
and report
\begin{equation}
\mathrm{Err}^{(k)} \triangleq F\!\left(\bar{x}^{(k)}\right),
\label{eq:wsloc_err}
\end{equation}
where $\bar{x}^{(k)}$ is decoded into $(\bar{p}_1^{(k)},\ldots,\bar{p}_{N_t}^{(k)})$. This produces a single scalar curve that directly reflects localization accuracy.
\begin{figure}[ht]
\vskip 0.2in
\begin{center}
\centerline{\includegraphics[width=\columnwidth]{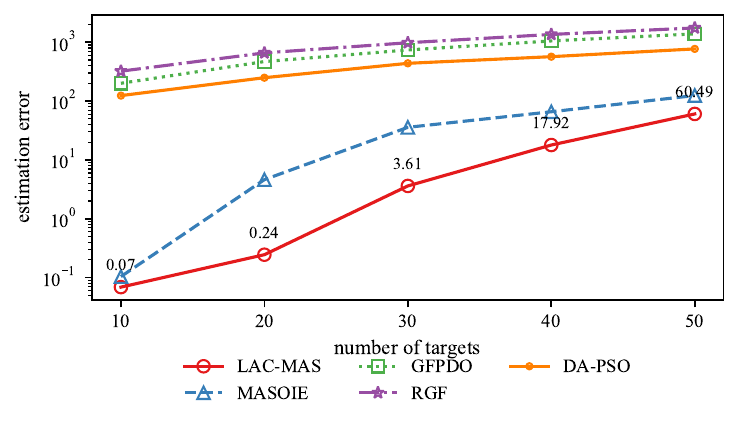}}
\caption{Performance of LAC-MAS and compared algorithms on the multi-target localization with different number of targets.}
\label{icml-historical}
\end{center}
\vskip -0.2in
\end{figure}

As shown in Fig.~3, LAC-MAS consistently achieves the lowest estimation error across all tested target numbers. Although the estimation error of all methods increases as $N_t$ grows, LAC-MAS maintains a substantially lower error level throughout, while the baseline methods converge to noticeably higher error plateaus. These results suggest that the proposed learning-to-act and learning-to-cooperate mechanism transfers beyond synthetic functions and remains effective in a realistic distributed black-box task under limited communication.

These results suggest that the proposed framework generalizes beyond synthetic functions and remains effective on structured distributed black-box tasks under limited communication. In this sense, the WSN task provides a representative transfer validation of the proposed trajectory-driven optimization framework in a realistic distributed setting.

\section{Conclusion}
This paper studied consensus-based black-box optimization from a learning perspective and proposed LAC-MAS, an LLM-assisted multi-agent framework that jointly learns how agents act and how they cooperate. By introducing adaptive regulation of agent-internal behaviors and agent-external coordination, and orchestrating their interaction through phased cognitive guidance, LAC-MAS effectively balances exploration, convergence, and communication efficiency.

Extensive benchmark experiments and ablation studies demonstrate that learning to act and learning to cooperate play complementary roles in distributed optimization. Internal behavioral learning improves solution quality and escape capability, while cooperative learning accelerates consensus formation and reduces communication cost. Their coordinated integration leads to stable and consistently strong performance across diverse problem landscapes.

\section*{Impact Statement}
This work focuses on improving the efficiency and robustness of distributed black-box optimization in multi-agent systems. Potential applications include cooperative sensing, resource allocation, and distributed control, which may contribute to more efficient and resilient large-scale systems. The proposed framework does not involve human subjects or personal data and is not expected to introduce significant ethical or societal risks beyond those common to general-purpose optimization technologies.


\bibliography{example_paper}
\bibliographystyle{icml2025}

\newpage
\appendix
\onecolumn
\section{Appendix A: Stage-wise Interpretation of Phased Cognitive Guidance}
The Phased Cognitive Guidance (PCG) mechanism introduced in Sec.~4.4 induces a small number of implicit stages along the optimization trajectory. These stages are not enforced as rigid phases, but emerge naturally from the interaction between the internal and external cognitive gates. They provide an interpretable abstraction for understanding how learning emphasis evolves over time.

\paragraph{Stage I: Trajectory Accumulation.}
During the initial phase, both internal and external cognitive gates remain inactive. Agents operate under the base optimization algorithm to explore the search space and collect foundational trajectories. This stage establishes the empirical basis required for subsequent trajectory-conditioned learning.

\paragraph{Stage II: Learning to Act.}
In this stage, internal behavioral learning becomes active while external cooperation learning remains suppressed. Agents leverage accumulated trajectory information to refine their internal action behaviors, allowing more effective local search after sufficient exploration has been achieved.

\paragraph{Stage III: Learning to Act and Cooperate.}
Both internal and external learning mechanisms are periodically activated. Agents simultaneously adapt their internal behaviors and re-evaluate cooperative interactions with neighbors, which enhances robustness and mitigates premature convergence caused by misleading local information.

\paragraph{Stage IV: Consensus-Oriented Cooperation.}
In the final stage, internal behavioral adaptation is deactivated, while external cooperation learning continues. Agents prioritize stable coordination and consensus formation, avoiding late-stage instability that may arise from excessive internal variation.

Overall, these stages provide an interpretable description of how PCG coordinates learning emphasis throughout the optimization process without imposing rigid temporal boundaries.

\section{Appendix B: Detailed Ablation Results}
\begin{table*}[htbp]
  \centering
  \caption{COMPARISON OF LAC-MAS WITH EXISTING ALGORITHMS}
  \label{tab:exact_match}
  \adjustbox{max width=\linewidth}{
\begin{tabular}{ccccccccccc}
\toprule[1pt] 
    \midrule[0.5pt]
   &   & MASOIE & LAC-MAS-Act & LAC-MAS-Coop & LAC-MAS &
   & MASOIE & LAC-MAS-Act & LAC-MAS-Coop & LAC-MAS  \\     
    \midrule[0.5pt]
fitness &   & 6.81e+04 & \textbf{3.61e+04} & \textbf{5.75e+04} & {\color[HTML]{FF0000}\textbf{2.21e+04}} &
   & 6.03e+06 & 7.91e+06 & {\color[HTML]{FF0000}\textbf{5.06e+06}} & 6.43e+06 \\
comm. cost &  & 1.28e+09 & 1.45e+09 & {\color[HTML]{FF0000}\textbf{7.85e+08}} & 1.37e+09 &
   & 4.48e+08 & 7.28e+08 & {\color[HTML]{FF0000}\textbf{2.68e+08}} & \textbf{3.71e+08} \\
improvement &   & - & 46.96\% & 15.56\% & 67.54\% &
   & - & -31.09\% & 16.20\% & -6.64\% \\
p-value & \multirow{-4}{*}{F1} & - & 1.18e-05\# & 2.28e-01 & 2.14e-13\# & 
\multirow{-4}{*}{F6}  & - & 1.96e-04\# & 1.00e-01 & 7.42e-01 \\
    \midrule[0.5pt]
fitness &   & 7.96e+04 & 8.65e+04 & 8.21e+04 & {\color[HTML]{FF0000} \textbf{7.57e+04}} &
& 1.47e+05 & \textbf{1.20e+05} & \textbf{1.31e+05} & {\color[HTML]{FF0000} \textbf{1.06e+05}} \\
comm. cost &   & 1.12e+08 & 1.71e+08 & {\color[HTML]{FF0000} \textbf{7.27e+07}} & \textbf{1.09e+08} &
& 3.63e+08 & 5.07e+08 & {\color[HTML]{FF0000} \textbf{2.13e+08}} & \textbf{3.22e+08} \\
improvement &   & - & -8.67\% & -3.16\% & 4.83\% &
& - & 18.22\% & 11.26\% & 28.19\% \\
p-value & \multirow{-4}{*}{F2} & - & 3.74e-02* & 3.81e-01 & 1.54e-01\# &
\multirow{-4}{*}{F7}  & - & 2.14e-02* & 1.25e-01 & 3.10e-03\# \\
    \midrule[0.5pt]
fitness &   & 1.05e+04 & \textbf{1.01e+04} & {\color[HTML]{FF0000} \textbf{9.85e+03}} & \textbf{9.99e+03} &
   & 4.13e+07 & {\color[HTML]{FF0000} \textbf{2.78e+07}} & 3.85e+07 & \textbf{3.03e+07} \\
comm. cost &   & 2.17e+08 & 3.81e+08 & {\color[HTML]{FF0000} \textbf{1.38e+08}} & \textbf{1.78e+08} &
   & 1.19e+09 & 1.78e+09 & {\color[HTML]{FF0000} \textbf{6.62e+08}} & \textbf{1.03e+09} \\
improvement &   & - & 3.28\% & 6.05\% & 4.75\% &
   & - & 32.83\% & 6.81\% & 26.74\% \\
p-value & \multirow{-4}{*}{F3} & - & 7.47E-01 & 7.47E-01 & 7.47E-01 &
\multirow{-4}{*}{F8}  & - & 1.03e-10\# & 1.00e-01 & 1.22e-08\# \\
    \midrule[0.5pt]
fitness &   & 2.42e+01 & \textbf{2.41e+01} & 2.43e+01 & {\color[HTML]{FF0000} \textbf{2.36e+01}} &
   & 1.28e+05 & 1.29e+05 & 1.39e+05 & {\color[HTML]{FF0000} \textbf{1.26E+05}} \\
comm. cost &   & 1.03e+08 & 1.72e+08 & {\color[HTML]{FF0000} \textbf{6.63e+07}} & 1.09e+08 &
   & 1.08e+08 & 1.71e+08 & {\color[HTML]{FF0000} \textbf{6.99e+07}} & 1.10e+08 \\
improvement &   & - & 0.08\% & -0.37\% & 2.42\% &
   & - & -0.85\% & -8.45\% & 1.52\% \\
p-value & \multirow{-4}{*}{F4} & - & 2.46e-01 & 2.46e-01 & 2.46e-01 &
\multirow{-4}{*}{F9}  & - & 1.00E+00 & 2.85e-02* & 2.73e-01 \\
    \midrule[0.5pt]
fitness &   & 2.14e+07 & \textbf{1.76e+07} & \textbf{2.13e+07} & {\color[HTML]{FF0000} \textbf{1.55e+07}} &
   & 8.28E+06 & \textbf{5.85e+06} & \textbf{6.33e+06} & {\color[HTML]{FF0000} \textbf{5.06E+06}} \\
comm. cost &   & 9.74e+08 & 1.26e+09 & {\color[HTML]{FF0000} \textbf{3.69e+08}} & \textbf{7.20e+08} &
   & 6.38e+08 & 2.29e+09 & {\color[HTML]{FF0000} \textbf{3.15e+08}} & 9.10e+08 \\
improvement &   & - & 17.93\% & 0.37\% & 27.59\% &
   & - & 29.29\% & 23.54\% & 38.82\% \\
p-value & \multirow{-4}{*}{F5} & - & 8.03e-05\# & 7.42e-01 & 2.31e-08\# &
\multirow{-4}{*}{F10} & - & 1.26e-04\# & 4.76e-02* & 7.98e-08\# \\
    \midrule[0.5pt]
    \bottomrule[1pt] 
\end{tabular}
 }
\vspace{2mm}
\begin{minipage}{1\linewidth}
\footnotesize
* and \# indicate statistical significance based on the Friedman test with Nemenyi post-hoc test at $\alpha=0.05$ and $\alpha=0.01$, respectively.
\end{minipage}
\end{table*}

\begin{figure*}[!http]
  \centering
  \subfigure[F5]{
    \includegraphics[width=0.45\linewidth]{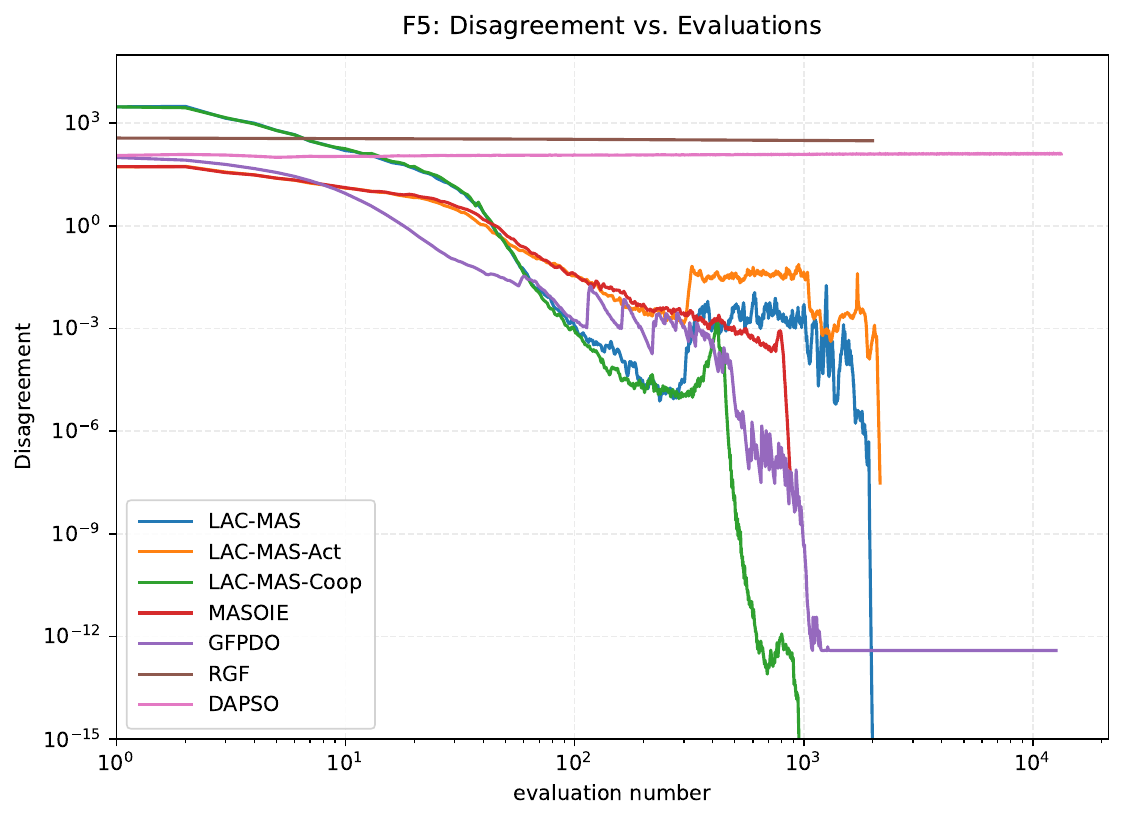}
  }
  \subfigure[F6]{
    \includegraphics[width=0.45\linewidth]{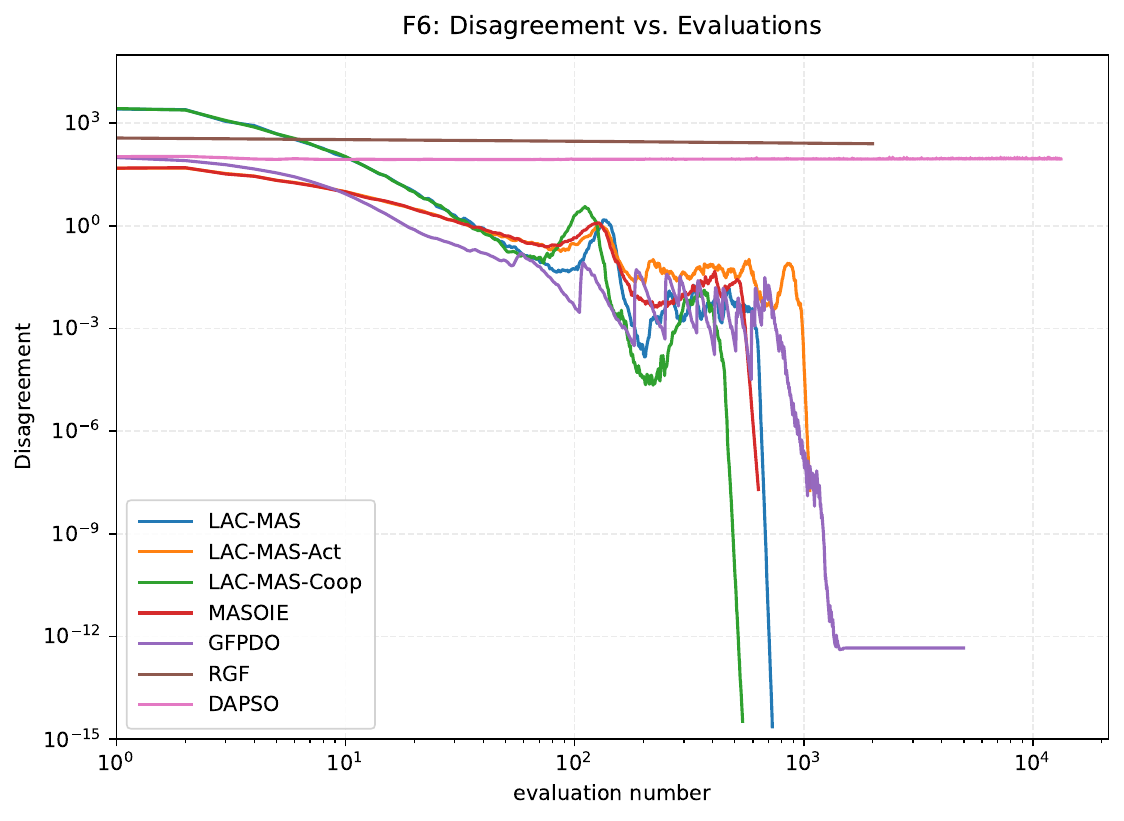}
  }
   \vspace{0.2em}
  \subfigure[F7]{
    \includegraphics[width=0.45\linewidth]{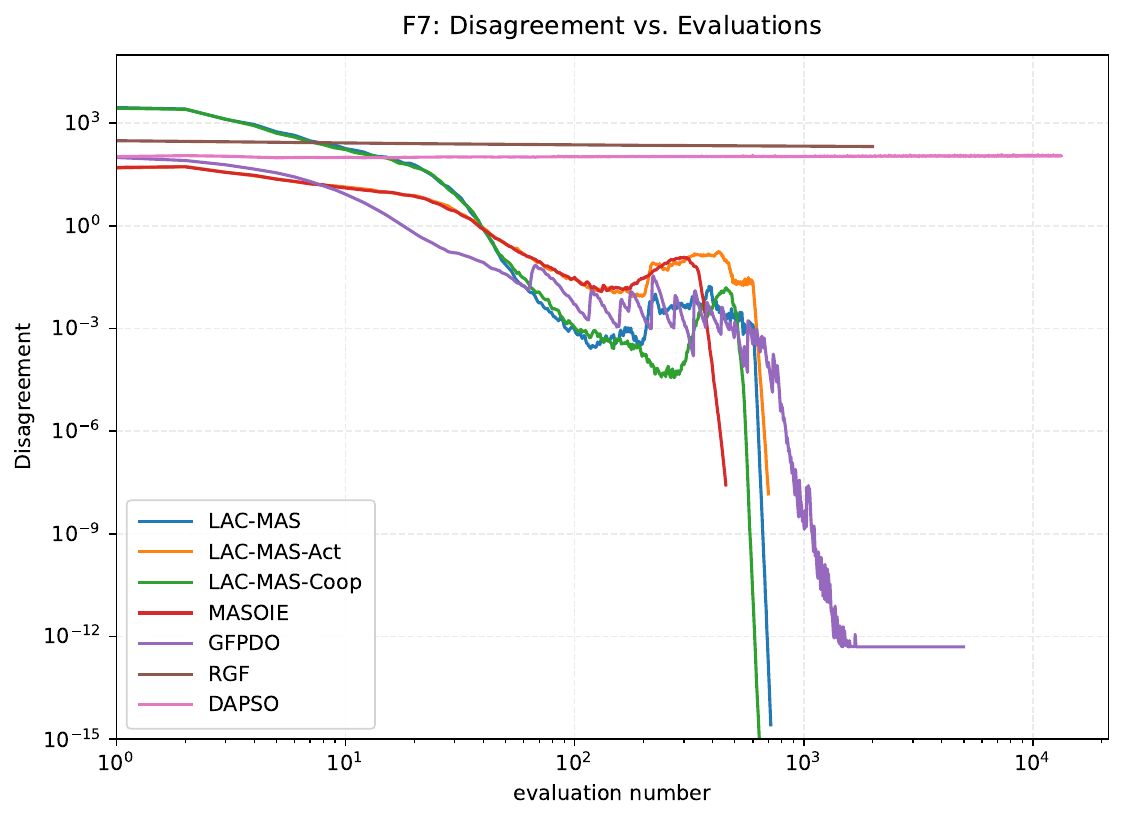}
  }
  \subfigure[F8]{
    \includegraphics[width=0.45\linewidth]{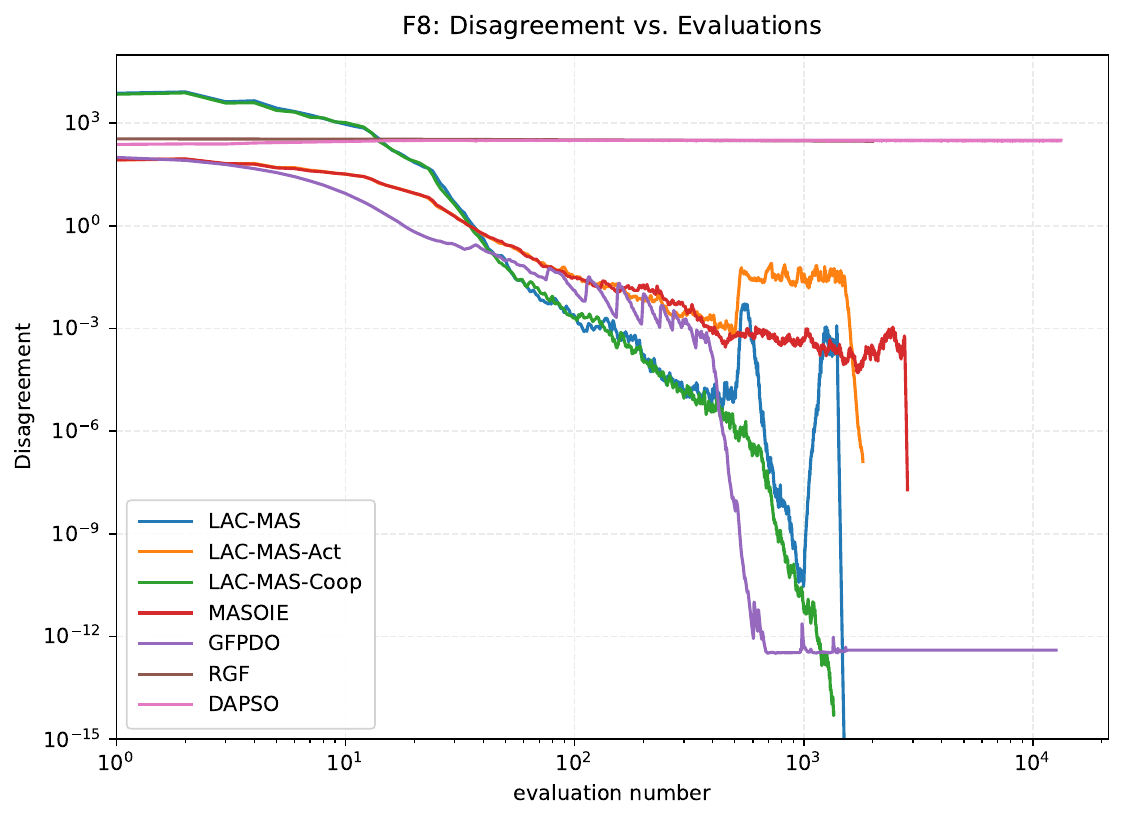}
  }
  \vspace{0.2em}
  \subfigure[F9]{
    \includegraphics[width=0.45\linewidth]{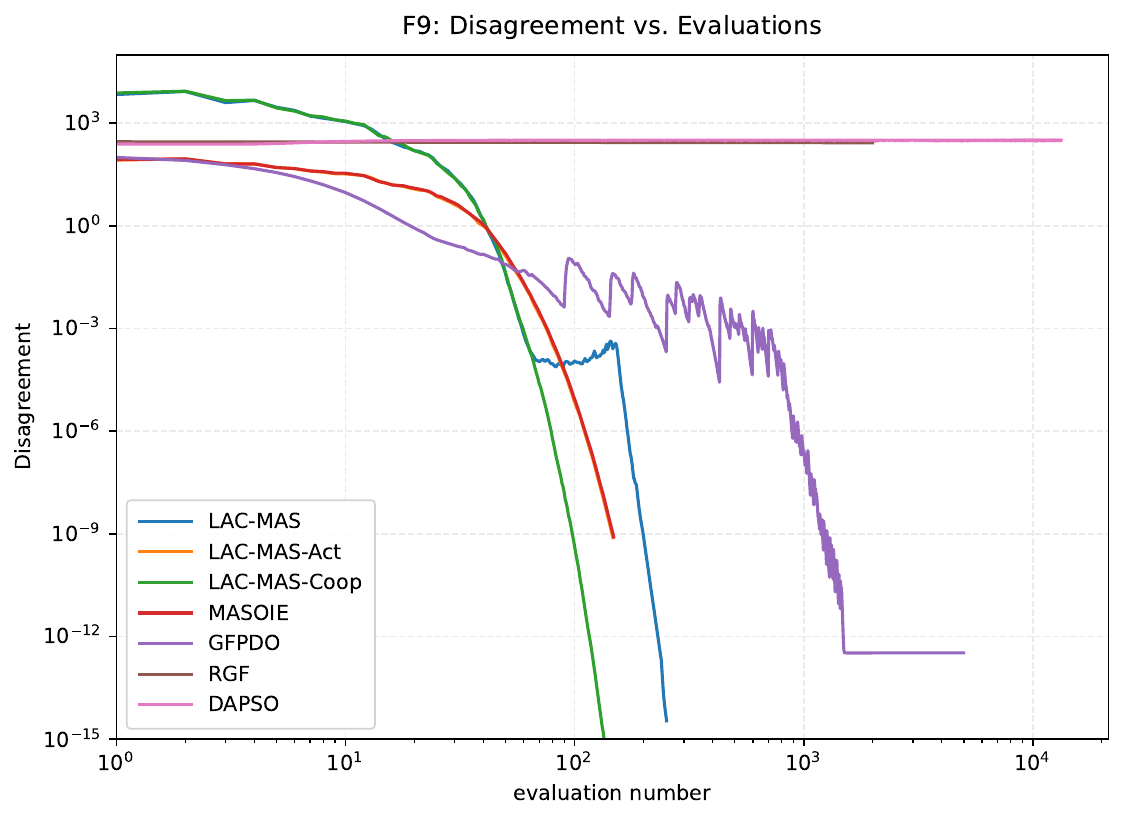}
  }
  \subfigure[F10]{
    \includegraphics[width=0.45\linewidth]{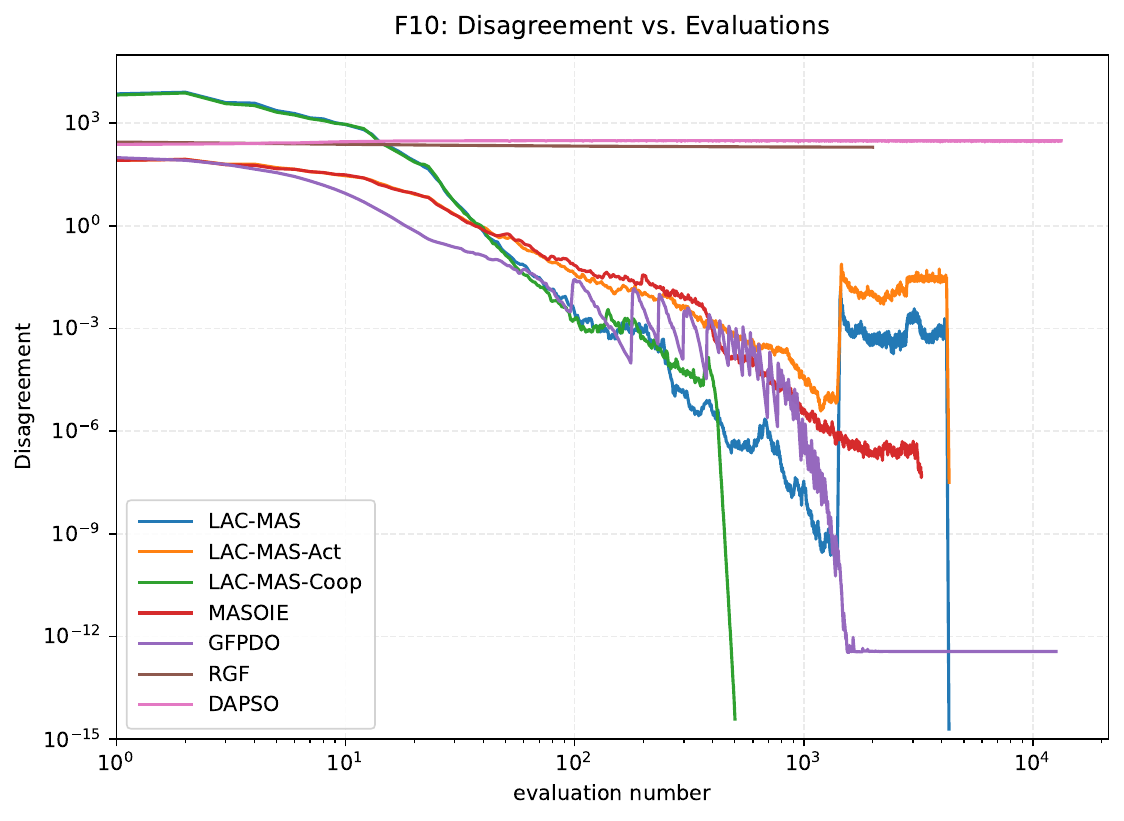}
  }
  \caption{
  Disagreement evolution on the remaining benchmark functions.}
  \label{fig:disagreement_all}
\end{figure*}

Table 2. presents detailed ablation results of LAC-MAS and its variants across ten benchmark functions. LAC-MAS-Act removes external cooperation learning, while LAC-MAS-Coop removes internal action learning, allowing the individual contributions of learning to act and learning to cooperate to be examined.

In the table, the communication cost (comm.\ cost) measures the cumulative communication overhead incurred during the optimization process, and red entries indicate the best performance for each metric on a given benchmark, while boldface values denote results that improve upon the baseline MASOIE. Lower fitness and communication cost correspond to better performance. Statistical significance is assessed using the Friedman test with the Nemenyi post-hoc test, where * and \# indicate significance levels at $\alpha=0.05$ and $\alpha=0.01$, respectively.

Overall, the full LAC-MAS achieves the best or near-best performance on most benchmarks, demonstrating the complementary effects of internal action learning and external cooperation learning. Specifically, variants equipped with internal action learning (Act) tend to achieve lower fitness values, indicating more effective local optimization, whereas variants with external cooperation learning (Coop) consistently reduce communication cost
by reallocating neighbor influence more efficiently. Models that include only one learning component therefore exhibit partial improvements, but still underperform the full LAC-MAS. These results confirm that jointly learning how to act and how to cooperate is essential for achieving robust and communication-efficient performance in black-box consensus optimization.

\section{Appendix C: Additional Implementation Details for Learning to Act and Cooperate}
\paragraph{Internal behavior modes.}
For each agent, the LLM outputs a small set of discrete internal behavior modes, which correspond to qualitatively different exploration--convergence regimes. These modes remain fixed within a short interval, providing stable trajectory-level guidance. At runtime, the agent selects an active mode based on recent divergence statistics computed over a short historical window, decoupling high-level behavioral learning from iteration-level responsiveness.
\paragraph{External cooperation weighting.}
For cooperation, each agent summarizes the recent behaviors of its neighbors using trajectory-based descriptors. The LLM maps these descriptors to relative importance scores, which are normalized to obtain adaptive consensus weights under a fixed and sparse communication topology. This design enables agents to emphasize informative neighbors while preserving stable information flow.
\paragraph{Stability and decentralization.}
All LLM-assisted outputs are applied in a decentralized manner and remain fixed within short intervals, preventing oscillatory behavior and preserving execution stability. No agent accesses global state or centralized supervision at any point. These design choices ensure that LAC-MAS remains fully decentralized while benefiting from trajectory-driven learning.

\section{Appendix D: LLM Configuration and Prompt Templates}
\label{app:llm_prompts}

This appendix provides the configuration details of the large language model (LLM) used in our experiments, together with representative prompt templates for learning to act and learning to cooperate. The purpose of this appendix is to clarify how trajectory information is encoded and queried, and to improve transparency and reproducibility of the proposed framework.

\subsection{LLM Configuration and Interaction Protocol}
All LLM-assisted components in LAC-MAS are implemented using a locally deployed large language model. Specifically, we adopt \texttt{DeepSeek-R1:14B}, deployed via the \texttt{Ollama} runtime, which enables fully local inference without reliance on external APIs or cloud services. This design supports local execution and is consistent with the decentralized implementation setting considered in this work.

The LLM is queried through lightweight HTTP requests using \texttt{curl}. Each query is executed independently by an agent based solely on locally available information, including its own historical optimization trajectories and statistics exchanged with neighboring agents. No global objective values, future information, or centralized state are accessible to the LLM at any time.

\subsection{Prompt for Learning to Act (Agent-Internal Behavior)}

For agent-internal behavior learning, the LLM is prompted to infer suitable internal action modes based on recent optimization trajectories. The prompt summarizes the current iteration, recent fitness and disagreement values over a fixed temporal window, and a concise adaptation rule describing the desired behavioral tendencies.

\begin{verbatim}
Tuning task: high-dimensional black-box optimization.
Current iteration: around <iter>.
Current parameters: d=<d>, c=<c>.
Recent trajectory (past 19 iterations):
Iteration <k>: fitness=<f_k>, disagreement=<g_k> |
...

Requirement:
If fitness stagnates while disagreement is low, increase c;
If fitness decreases slowly while disagreement is high, increase d.
Only return the updated parameters in parentheses, separated by a comma.
Constraints: d in [0.5, 1], c in [1, 1.8].
Example: (0.7, 1.3)
\end{verbatim}

\subsection{Prompt for Learning to Cooperate (Agent-External Coordination)}

For learning cooperative behaviors, the LLM is prompted to adapt neighbor influence weights based on aggregated historical statistics of neighboring agents. Each prompt encodes the number of neighbors, recent average fitness and disagreement values, and normalization constraints.

\begin{verbatim}
Task: update the neighbor weight vector for multi-agent optimization.
Number of neighbors: <N>.

Weight update rules:
1. If a neighbor has low fitness and low disagreement, increase its weight (0.3–0.5);
2. If a neighbor has high fitness and high disagreement, decrease its weight (0.1–0.2);
3. Fitness is prioritized; weights must sum to 1.

Neighbor performance history (last 10 iterations):
Neighbor ID <i>: avg fitness=<f_i>, avg disagreement=<g_i> |
...

Please return the updated weights in the format [w1, w2, ..., wN].
\end{verbatim}

\paragraph{Remark.}
The prompt templates above serve as structured interfaces between trajectory statistics and adaptive decision-making. They are intentionally lightweight and generic, and do not assume access to explicit objective models, benchmark identifiers, or centralized coordination signals. As such, they remain aligned with the decentralized black-box consensus optimization setting considered in this work.

\section{Appendix E. Additional Details for Consensus Preservation}

This appendix provides additional details for the consensus claim in Theorem~4.1. Our purpose is not to re-prove a full classical consensus theorem from first principles, but to verify that the closed-loop dynamics of LAC-MAS satisfy the admissibility conditions required by standard consensus results for connected row-stochastic switching systems with asymptotically vanishing perturbations.

\paragraph{Agent-level consensus form.}
Let $z_i^{(t)}$ denote the agent-level representative state used in consensus fusion at iteration $t$, and define the stacked vector
\begin{equation}
z^{(t)} = [z_1^{(t)}, \ldots, z_N^{(t)}]^\top.
\end{equation}
Then the cooperative update of LAC-MAS can be written as
\begin{equation}
z^{(t+1)} = A^{(t)} z^{(t)} + \xi^{(t)},
\label{eq:appendix_consensus_form}
\end{equation}
where $A^{(t)} = [a_{ik}^{(t)}]$ is the mixing matrix induced by the cooperation weights, and $\xi^{(t)}$ collects the effective perturbation introduced by local swarm evolution and internal adaptive execution before consensus fusion.

\paragraph{Lemma C.1 (Admissibility of the cooperation matrix).}
For every iteration $t$, the matrix $A^{(t)}$ is nonnegative, graph-compatible, and row-stochastic, i.e.,
\begin{equation}
\begin{aligned}
&a_{ik}^{(t)} \ge 0,\\
&a_{ik}^{(t)} = 0 \quad \text{if } k \notin \mathcal{N}_i \cup \{i\},\\
&\sum_{k \in \mathcal{N}_i \cup \{i\}} a_{ik}^{(t)} = 1.
\end{aligned}
\end{equation}

\textit{Justification.}
By construction, the LLM only assigns candidate weights over the existing neighbor set $\mathcal{N}_i \cup \{i\}$. The subsequent normalization/projection step guarantees nonnegativity and unit row sum before execution. Hence $A^{(t)}$ is row-stochastic and preserves the original communication sparsity pattern. Since the communication graph is fixed and connected, the induced switching family $\{A^{(t)}\}$ remains compatible with the same connected graph.

\paragraph{Lemma C.2 (Bounded finite-stage internal adaptation).}
The internal action mechanism introduces only bounded and finitely refreshed modulation into the low-level swarm dynamics.

\textit{Justification.}
The internal coefficients are selected from the finite set
\(
\mathbf{w}_i = (w_{i,1}, w_0, w_{i,2})
\),
hence they are bounded. The modulation vectors $\Delta_{i,p}^{(t)}$ are also bounded by assumption. Therefore, the internal velocity update remains a bounded modulation of the underlying swarm dynamics. Moreover, under PCG, internal-guidance refresh occurs only at the finitely many scheduled times in $\mathcal{T}_{\mathrm{int}}$, and is deactivated after the calibrated horizon $T$. Thus the internal adaptation does not induce persistent high-frequency switching.

\paragraph{Lemma C.3 (Asymptotically vanishing perturbation).}
In the late-stage stable regime induced by PCG, the effective perturbation term in \eqref{eq:appendix_consensus_form} satisfies
\begin{equation}
\|\xi^{(t)}\| \to 0, \qquad t \to \infty.
\end{equation}

\textit{Justification.}
After the final internal-guidance refresh, the agent-level execution enters a stable regime in which internal coefficient switching ceases. In this regime, the remaining variation entering consensus fusion is caused only by the decaying local swarm adjustment around the stabilized execution dynamics. Therefore, the perturbation term contributed by local black-box search to the agent-level consensus update vanishes asymptotically.

\paragraph{Lemma C.4 (Consensus contraction under admissible switching).}
Consider the disagreement projector
\begin{equation}
J = I - \frac{1}{N}\mathbf{1}\mathbf{1}^\top.
\end{equation}
Applying $J$ to \eqref{eq:appendix_consensus_form} gives
\begin{equation}
J z^{(t+1)} = J A^{(t)} z^{(t)} + J \xi^{(t)}.
\end{equation}
Because each $A^{(t)}$ is row-stochastic and graph-compatible on a connected graph, the consensus subspace $\mathrm{span}\{\mathbf{1}\}$ is invariant, and the disagreement component is contracted under the associated switching consensus dynamics. Since $\xi^{(t)} \to 0$, the disagreement term asymptotically vanishes.

\paragraph{Proof of Theorem~4.1.}
By Lemma C.1, the cooperation matrix $A^{(t)}$ remains admissible for all $t$. By Lemma C.2, the internal action mechanism introduces only bounded finite-stage switching. By Lemma C.3, the perturbation term vanishes asymptotically in the late-stage stable regime. Therefore, the closed-loop system \eqref{eq:appendix_consensus_form} is a connected row-stochastic switching consensus system with asymptotically vanishing perturbations. Standard consensus arguments for such systems imply
\begin{equation}
\|z_i^{(t)} - z_j^{(t)}\| \to 0, \qquad \forall i,j.
\end{equation}
Since $z_i^{(t)}$ is the agent-level representative state used in consensus fusion, this yields the claimed consensus result in Theorem~4.1. \hfill$\square$

\section{Appendix F. Symbol Table}

\renewcommand{\arraystretch}{1.1}
\begin{longtable}{p{0.18\textwidth} p{0.74\textwidth}}
\caption{Main symbols used in LAC-MAS.}\label{tab:symbol_table}\\
\toprule
\textbf{Symbol} & \textbf{Meaning} \\
\midrule
\endfirsthead

\toprule
\textbf{Symbol} & \textbf{Meaning} \\
\midrule
\endhead

\midrule
\multicolumn{2}{r}{Continued on next page} \\
\endfoot

\bottomrule
\endlastfoot

$\mathcal{G}=(\mathcal{V},\mathcal{E})$ & Fixed connected communication graph of the multi-agent system. \\

$\mathcal{V}=\{1,\dots,N\}$ & Set of agents (nodes) in the distributed system. \\

$\mathcal{N}_i$ & Neighbor set of agent $i$ in the communication graph. \\

$N$ & Number of agents. \\

$D$ & Dimension of the decision space. \\

$f_i:\mathbb{R}^D\to\mathbb{R}$ & Local black-box objective function of agent $i$. \\

$f(x)=\frac{1}{N}\sum_{i=1}^N f_i(x)$ & Global objective defined as the average of local objectives. \\

$x_i$ & Agent-level decision variable of agent $i$ in the distributed consensus formulation. \\

$\{x_{i,p}^{(t)}\}_{p=1}^P$ & Local particle population maintained by agent $i$ at iteration $t$. \\

$x_{i,p}^{(t)}\in\mathbb{R}^D$ & Position of particle $p$ in agent $i$ at iteration $t$. \\

$v_{i,p}^{(t)}\in\mathbb{R}^D$ & Velocity of particle $p$ in agent $i$ at iteration $t$. \\

$P$ & Population size of each local swarm optimizer. \\

$\mu_i^{(t)}$ & Centroid of the local particle population of agent $i$ at iteration $t$. \\

$D_i^{(t)}$ & Particle divergence of agent $i$, measuring the dispersion of its local particle population. \\

$\mathbf{w}_i=(w_{i,1},w_0,w_{i,2})$ & Internal behavioral coefficient set for agent $i$. \\

$w_i^{(t)}$ & Active internal coefficient selected according to the current divergence regime. \\

$d_1,d_2$ & Divergence thresholds with $d_1<d_2$. \\

$\Delta_{i,p}^{(t)}$ & Random modulation vector generated by the underlying swarm update rule. \\

$\odot$ & Element-wise multiplication operator. \\

$\mathcal{H}_i^{(t)}$ & Local trajectory/history information maintained by agent $i$ up to iteration $t$. \\

$\mathbf{s}_{ik}^{(t)}$ & Trajectory-based descriptor used by agent $i$ to evaluate neighbor $k$. \\

$\bar f_k^{(t)}$ & Recent average objective value of neighbor $k$. \\

$\bar D_k^{(t)}$ & Recent average particle divergence of neighbor $k$. \\

$\overline{\|\Delta x_k\|}^{(t)}$ & Recent average magnitude of state variation of neighbor $k$. \\

$a_{ik}^{(t)}$ & Cooperation weight assigned by agent $i$ to neighbor $k$ (or itself) at iteration $t$. \\

$A^{(t)}=[a_{ik}^{(t)}]$ & Time-varying row-stochastic mixing matrix induced by adaptive cooperation weights. \\

$g_{\mathrm{ext}}(t)$ & Cooperation-refresh gate. \\

$g_{\mathrm{int}}(t)$ & Internal-action refresh gate. \\

$\mathcal{T}_{\mathrm{ext}}$ & Set of iterations at which cooperation guidance is refreshed. \\

$\mathcal{T}_{\mathrm{int}}$ & Set of iterations at which internal action guidance is refreshed. \\

$\rho_{\mathrm{ext}}$ & Refresh interval ratio for cooperation-guidance updates. \\

$\rho_1,\rho_2$ & Two key refresh ratios for internal action guidance. \\

$T$ & Characteristic optimization horizon estimated by pre-experiment calibration in PCG. \\

$s(t)$ & Implicit stage index induced by the interaction of the two cognitive gates. \\

$\tau_m=\lceil \alpha_m T\rceil$ & Stage transition point used in the stage-wise interpretation of PCG. \\

$x_i^{(t)}$ & Agent-level representative state of agent $i$ used in consensus fusion at iteration $t$. \\

$x^{(t)}=[x_1^{(t)},\dots,x_N^{(t)}]^\top$ & Stacked vector of agent-level representative states. \\

$\xi^{(t)}$ & Effective perturbation term in the closed-loop consensus dynamics. \\

$p_t\in\mathbb{R}^3$ & Position of the $t$-th target in the WSN localization task. \\

$y_i\in\mathbb{R}^3$ & Known location of sensor $i$ in the WSN localization task. \\

$N_t$ & Number of targets in the WSN localization task. \\

$\phi_{it}$ & RSS measurement from target $t$ to sensor $i$. \\

$P_0$ & Reference RSS value at distance $d_0$. \\

$d_0$ & Reference distance in the path-loss model. \\

$n_p$ & Path-loss exponent in the WSN localization model. \\

$\bar{x}^{(k)}=\frac{1}{n}\sum_{i=1}^n x_i^{(k)}$ & System-level estimate averaged across all sensors at communication round $k$. \\

$\mathrm{Err}^{(k)}$ & Estimation error used in the WSN task, defined as $F(\bar{x}^{(k)})$. \\
\end{longtable}


\end{document}